\documentclass[11pt]{article}
\usepackage{bm, commath}
\usepackage[authoryear]{natbib}
\usepackage{caption}
\usepackage{subcaption}
\usepackage{graphicx}
\usepackage{amsmath, amsfonts, amsthm}
\usepackage{float}
\usepackage{booktabs,siunitx}
\usepackage{url}
\usepackage{multirow}
\usepackage{amssymb}
\usepackage{blkarray}
\usepackage{bbm}
\usepackage{multicol}
\usepackage{authblk}
\usepackage[flushleft]{threeparttable}
\usepackage{enumitem}
\usepackage[nottoc]{tocbibind}
\usepackage{comment} % commenting 
\usepackage{mathrsfs}
\usepackage{amsmath}
\usepackage{algorithm}
\usepackage{algpseudocode}
\usepackage{mathtools} 
\usepackage[margin= 0.9in]{geometry} % 1 inch margin
\usepackage{setspace} 

\usepackage{comment}
\usepackage{prodint}
\usepackage{xr}

\usepackage{listings}
\usepackage{bbm}
\usepackage{textcomp}
\usepackage{sectsty}
\usepackage[utf8]{inputenc}
\usepackage{float}
\usepackage{tikz}
\usetikzlibrary{shapes,decorations,arrows,calc,arrows.meta,fit,positioning}

\usepackage[colorlinks=true,
            linkcolor=blue,
            citecolor=blue,
            urlcolor=blue]{hyperref}

\makeatletter
\newcommand{\setcustomref}[2]{\def\@currentlabel{#2}\label{#1}}
\makeatother

\newtheorem{theorem}{Theorem}

\newtheorem{remark}{Remark}

\newtheorem*{assumption*}{Assumption}

\newtheorem{proposition}{Proposition}

\newtheorem{lemma}{Lemma}

\newtheorem{innercustomgeneric}{\customgenericname}
\providecommand{\customgenericname}{}
\newcommand{\newcustomtheorem}[2]{%
  \newenvironment{#1}[1]
  {%
   \renewcommand\customgenericname{#2}%
   \renewcommand\theinnercustomgeneric{##1}%
   \innercustomgeneric
  }
  {\endinnercustomgeneric}
}

\newcustomtheorem{customthm}{Theorem}
\newcustomtheorem{customassumption}{Assumption}

\usepackage{mathtools}

\usepackage{xcolor}

\newcommand{\E}{\mathbb{E}}

\newcommand{\Pp}{\mathbb{P}}

\begin{document}

\def\spacingset#1{\renewcommand{\baselinestretch}%
{#1}\small\normalsize} \spacingset{1.5}

\linespread{1.5}

\doublespacing

\sectionfont{\bfseries\large\sffamily}%
\subsectionfont{\bfseries\sffamily\normalsize}%

\title{A general nonparametric framework for testing hypotheses about function-valued parameters}
\author[1]{Albert Osom \thanks{Corresponding author: {\tt aosom@uw.edu}}}
\author[1,2]{Ali Shojaie}
\author[1,3]{Aaron Hudson }
\affil[1]{Department of Biostatistics, University of Washington}
\affil[2]{Department of Statistics, University of Washington}
\affil[3]{Vaccine and Infectious Disease Division, Fred Hutchinson Cancer Center}
\date{}

\maketitle

\begin{abstract}
\noindent We present a general nonparametric approach for testing whether a statistical parameter defined through conditional distributions is constant across the conditioning variables. Such hypotheses arise naturally in problems such as assessing treatment effect heterogeneity, conditional associational effects, and conditional mean dependence. Our framework studies function-valued parameters obtained by evaluating a smooth statistical functional on conditional probability distributions. We establish an explicit connection between our test and procedures based on studying the norm of the function-valued parameter. Unlike many existing norm-based tests, which  exhibit poor asymptotic behavior under the null, the proposed test statistic admits a tractable limiting null distribution. We illustrate the applicability of the proposed test through several examples, assess its operating characteristics in simulation studies, and apply it to data from a breast cancer trial to identify predictive biomarkers for response to adjuvant chemotherapy.
\end{abstract}

\section{Introduction}\label{intro}
Testing hypotheses about function-valued parameters, such as conditional mean and conditional average treatment effect (CATE) functions, arises commonly in many scientific applications, and it presents a fundamental challenge, particularly in nonparametric models. For instance, clinical investigators often seek to determine whether the effect of an intervention on a health outcome differs across subgroups defined by patient characteristics, or if certain patient characteristics predict a health outcome.  Addressing these questions may require performing inference on CATE or a conditional mean function. Both of these parameters, which we will study later, fall within a broader class of function-valued parameters that can be expressed as an evaluation of a smooth statistical functional on a conditional probability distribution.
The current paper focuses on this class of function-valued parameters, which encompasses many other parameters, including the conditional covariance function, the conditional survival function, and continuous linear functionals of the conditional mean function \citep{chernozhukov2024conditional}. We develop a general framework for testing whether the evaluation of a given smooth statistical functional on a conditional probability distribution is constant in the conditioning variables. For example, this null can correspond with equality of marginal and conditional means in a mean regression problem.

%%%%%%%%%
\subsection{Related literature}
%%%%%%%%%
Inference on function-valued parameters is a long-standing research topic, but methodological work remains limited, especially when aiming to remain agnostic to the data-generating mechanism. In contrast, if one restricts the space of distributions that generate the data (e.g., parametric models) or the structure of the function-valued parameters (e.g., linearity), inference becomes much simpler and has been well studied.  For instance, to test whether a variable is a significant predictor of the outcome in a linear regression model, it suffices to test whether its coefficient in the model is zero. Although this approach is straightforward, it may lead to misleading conclusions if the model assumptions are incorrect or misspecified \citep{whitney2020comment}. It is thus desirable to relax the restrictions on the model space and allow the function-valued parameter to fall within a nonparametric or semi-parametric model. Inference under this relaxed model space is challenging, as it often requires flexibly estimating the function-valued parameter using machine learning or data-adaptive approaches, which typically do not converge at the standard parametric $n^{1/2}$-rate. General inferential strategies have nonetheless been proposed for inference on certain classes of function-valued parameters under nonparametric models.  %\citep{luedtke2019omnibus,hudson2023nonparametric,westling2022nonparametric}.

Existing hypothesis tests for function-valued parameters, including \cite{luedtke2019omnibus}, \cite{hudson2023nonparametric}, and \cite{westling2022nonparametric},  rely on test statistics that are based on a one-dimensional summary of the function-valued parameter. In particular, \cite{westling2022nonparametric} introduce a test statistic based on an estimator of the primitive (i.e., integrated) form of the causal dose-response function to assess whether it is constant. Their approach can be used to test hypotheses about other function-valued parameters \citep[see , e.g.,][]{wang2024nested}, but the test may suffer from low power in finite samples when the function-valued parameter is complex or non-monotone. 
\cite{hudson2023approach} and \cite{jin2024class} develop generalized tests for constancy of the dose-response function that can adapt to non-monotonicity.
\cite{hudson2023nonparametric} focus on function-valued parameters that can be expressed as the minimizer of a goodness-of-fit measure (e.g., expected squared error). They construct their test statistic based on the maximum difference between the estimator of the minimum goodness-of-fit measure and the goodness-of-fit measure evaluated at the null parameter along a path that passes through the function-valued parameter. However, if the function-valued parameter of interest is not a minimizer of some goodness-of-fit measure, their approach is inapplicable. \cite{luedtke2019omnibus} propose a test statistic based on the mean discrepancy \citep{gretton2006kernel} between the distributions of the function-valued parameter and the null parameter, maximized within the unit reproducing kernel Hilbert space (RKHS). Although this framework is more broadly applicable, it requires second-order pathwise differentiability \citep{robins2008higher, robins2009quadratic, van2014higher} to avoid degeneracy issues in characterizing the limiting null distribution of the test statistic. Even when second-order pathwise differentiability holds, obtaining second-order influence functions remains challenging \citep{fisher2021visually}.

\subsection{Our contributions}
We propose a nonparametric test for a class of null hypotheses concerning the constancy of a function-valued parameter, defined as the evaluation of a smooth functional at a conditional distribution. Our approach develops a norm-based test by exploiting the representation of the norm as a supremum over finite-dimensional parameters indexed by functions in a prescribed class. These finite-dimensional parameters are estimated using debiased machine-learning techniques \citep{chernozhukov2018double} that attain $n^{1/2}$-rate of convergence without requiring sample splitting, and their supremum serves as the test statistic.
We show that our test statistic has a well-behaved limiting null distribution and admits an asymptotic representation based on influence functions. Moreover, using the the supremum representation of the norm allows us to choose function classes that leverage structural assumptions on the function-valued parameter for more efficient estimation of the norm. Our test is guaranteed to achieve asymptotic type-I error control, and its power against a given alternative depends on the choice of the function class. To ensure power against a wide range of alternatives, we propose an aggregate test that combines several function classes. We also describe a broad class of function-valued parameters for which our proposed test statistic can be constructed without explicitly deriving influence functions from scratch. This feature makes the framework easy to specialize to additional examples without requiring as much specialized knowledge about semi-/non-parametric efficiency theory.

Our approach has the following advantages over existing methods. Firstly,  our test can be more powerful than tests constructed based on direct estimation of the norm of the function-valued parameter. Such tests are consistent since the norm is zero if and only if the null is true. However, direct estimation of the norm does not admit a tractable limiting distribution under the null \citep{williamson2021nonparametric}, making it difficult to obtain critical values that guarantee type-I error control. Though one can obtain well-behaved norm-estimators using sample-splitting \citep{williamson2023general,dai2022significance,cai2024test}, this  approach results in low power due to inefficient use of the data. Our test avoids these issues and establishes a theoretical size guarantee, without the need for sample-splitting. Secondly, our framework allows seemless construction of test statistics that can leverage known or assumed structure on the function-valued parameter---e.g., monotonicity, sparsity, and smoothness---to improve power, using a similar strategy as described in \cite{jin2024class}. Thirdly, we introduce an aggregate test that combines different individual test statistics, while preserving tight size control. Lastly, our proposed test is easy to implement and readily generalizes to new testing problems. In particular, unlike the method of \cite{luedtke2019omnibus}, our method requires only first-order pathwise derivatives, and we provide theoretical results that greatly simplify these derivations, eliminating the need for manual calculations in many settings. Our test also applies to classes of function-valued parameters that are not covered by the method of \cite{hudson2023approach} and are not addressed in \cite{luedtke2019omnibus}.

The remainder of the paper is organized as follows. Section~\ref{examples} presents the general testing problem and introduces two main illustrative examples: testing conditional mean dependence and testing treatment effect heterogeneity. Section~\ref{methodology} describes our testing strategy and the proposed test statistic. We apply our framework to the illustrative examples in Section~\ref{illustration}. In Section~\ref{aggregatetest}, we examine some choices of function classes over which the supremum can be evaluated and describe an aggregate test for combining multiple test statistics. Section~\ref{simulation} evaluates the performance of our method through simulation studies, and Section~\ref{application} presents a data example. Section~\ref{conclusion} provides concluding remarks and directions for future research.

\section{Overview of the testing problem and examples \label{examples}}

Let $O=(U,V)$ be a random vector drawn from a distribution $P_0$ with support $\mathcal{O}=\mathcal{U} \times \mathcal{V}$, where $\mathcal{U}$ and $\mathcal{V}$ are the supports of $U$ and $V$, respectively.
Suppose we observe  $\{ O_i=(U_i,V_i)\}^{n}_{i=1}$ independent and identically distributed (i.i.d.) copies of $O=(U,V)$ from the distribution $P_0$; let $\mathbb{P}_n$ denote the corresponding empirical distribution. We assume $P_0\in \mathcal{M}$, where $\mathcal{M}$ is the collection of all distributions on $\mathcal{O}$. Let $P_{0,U}$, $P_{0,V}$, $P_{0,U|V}$, and $P_{0,O|V}$ denote the respective  marginal and conditional distributions derived from $P_{0}$. We use $o= (u,v)$ to denote a realization of the random vector $O = (U,V)$. For $V=v$ fixed, we denote $P_{0,O|v}$ as the conditional distribution of $O$ given $V=v$. For any $P \in \mathcal{M}$, let $\Psi(P)$ denote a smooth functional; that is, we assume $\Psi$ is pathwise differentiable at $P_0$ \citep{bickel1993efficient}. Our goal is to make inference on the corresponding function-valued parameter obtained by evaluating $\Psi$ at the conditional distribution of $O$ given $V=v$. For notational simplicity, we write $\Psi_0 = \Psi(P_0)$ and $\Psi_{0,v} = \Psi(P_{0,O|v})$. For example, we can consider $\Psi_0 = \int u \, dP_{0}(o)$ and $\Psi_{0,v} = \int u \, dP_{0,O|v}(o)$. We define the class of $P_{0}\; p$-integrable functions with induced norm $\|f\|_{L^{p}(P_{0})}$  as $L^{p}(P_{0})\coloneq \{ f : \mathcal{O} \rightarrow \mathbb{R} : \|f\|_{L^{p}(P_{0})} < \infty  \}$, where
$\|f\|_{L^{p}(P_{0})}\coloneq \Big [ \int |f|^p dP_{0}\Big]^{1/p}$ for  $p \in [1,\infty)$ and  $ \|f\|_ {L^{\infty}(P_{0})}\coloneq \underset{v}{\text{ess sup}}|f(v)|$.  Denoting $\theta_{0}=\int \Psi_{0,v}(v)dP_{0,V}(v)$, our goal is to test hypotheses of the form:
\begin{align}
    H_0:\Psi_{0,v}=\theta_{0} \quad  \quad \text{versus} \quad  \quad H_1: \Psi_{0,v} \ne \theta_{0}, \quad P_{0,V}\text{-almost surely}. \label{nullhypothesis} 
\end{align}
A special case of \eqref{nullhypothesis} is when $\theta_{0} \equiv 0$. Before presenting our framework for testing  \eqref{nullhypothesis}, we present two motivating examples that fall within the class of inference problems of form \eqref{nullhypothesis}: testing for conditional mean dependence and assessing the presence of treatment effect heterogeneity. In the Supplementary Material \ref{sec:additional-examples}, we also include hypotheses tests about the conditional covariance function and conditional survival function as additional examples.

\subsection*{Example 1 \normalfont{(Testing for conditional mean dependence)}} 
\setcustomref{examp1}{1}  % Manually set reference to 3
A fundamental task in predictive modeling is to determine whether a set of covariates can significantly improve the prediction of an  outcome of interest in a mean-squared error sense. This is typically assessed by studying the relationship between the outcome and covariates using the conditional mean function and determining whether the conditional mean function depends on the given set of covariates. When the conditional mean function is correctly specified as a linear function of the covariates, standard techniques such as the $F$-test provide valid and powerful tests for the significance of covariates. However, when the linear model is misspecified, type-I error rate may not be controlled and power can be substantially reduced. To overcome challenges associated with model misspecification, \cite{zhang2018conditional} and \cite{lai2021kernel} proposed model-free tests for evaluating the null hypothesis of conditional mean independence. These methods improve power by formulating conditional mean independence as a global norm-based hypothesis on a functional object, enabling detection of broad classes of nonlinear alternatives. Despite these advances, overall conditional mean independence testing remains an active area of research, particularly in high-dimensional and nonparametric settings \citep{hudson2026inference,li2023testing,williamson2023general,cai2024test}. We show that hypothesis testing for conditional mean dependence fits within the general framework described in~\eqref{nullhypothesis}.

Let $O=(X,Y) \sim P_{0}$, where $Y\in \mathbb{R}$ is the outcome, or response, of interest, assumed to be bounded, and $X\in \mathbb{R}^{p}$ is a set of covariates. Defining our smooth map as $\Psi(P)= \int y dP $ for $P \in \mathcal{M}$, the function-valued and null parameters in \eqref{nullhypothesis} are $\Psi_{0,x}=\E[Y|X=x]$ and $\theta_0 = \E[Y]$, respectively.
Then, the null hypothesis of conditional mean independence, can be evaluated by testing whetehr the conditional mean of $Y$ given $X$ depend on $X$. Specifically, %the null hypothesis of interest is:
\begin{align*}
    H_0: \E[Y|X=x]=\E[Y] \quad P_{0,X}\text{-almost surely}.
\end{align*}
The above hypothesis clearly falls within our general class of null hypotheses in \eqref{nullhypothesis}. 

\subsection*{Example 2 \normalfont{(Assessing treatment effect heterogeneity)}}
\setcustomref{examp2}{2}  % Manually set reference to 2

Evaluating whether responses to an intervention are heterogeneous across a population has been of interest for many years \citep[see][for early discussions]{byar1985assessing}. When heterogeneity arises among subgroups defined by discrete covariates (e.g., sex, or age groups), several parametric \citep{byar1985assessing, gail1985testing} and nonparametric \citep{imai2024statistical, dai2023nonparametric, chernozhukov2017fisher, crump2008nonparametric} tests for treatment effect heterogeneity have been proposed, mainly for randomized trials, with some extensions to observational studies. These methods assess effect heterogeneity by testing the constancy of the conditional average treatment effect, where the conditioning covariate corresponds to the subgroup-defining variable. However, in many settings, investigators do not have clearly defined subgroups,  but instead observe continuous or multivariate covariates that may capture the heterogeneity. Ad-hoc strategies, such as discretizing continuous covariates and applying the existing methods often yield low-powered tests, especially in small samples.

To test the null hypothesis of treatment effect homogeneity, we also propose evaluating the constancy of the CATE. Specifically, we test the null hypothesis that the CATE does not depend on the conditioning variable. Our proposed test does not require pre-defined subgroups and works even when the conditioning covariate is continuous or multivariate.

Let $O=(Y^{(1)},Y^{(0)},T,X) \sim P_0$ be the data we wish to have observed, where $Y^{(t)}$ denotes a binary or continuous  potential outcome under treatment $T=t$ and $X$ denotes the baseline covariates. Let $T\in \{0,1\}$  and for $s \subset \{1,\ldots,p\}$, let $X_s$ be the sub-vector containing elements of $X$ with indices in $s$. For this example, our smooth map is $\Psi(P) =\int \left(Y^{(1)}-Y^{(0)}\right)dP(o)$ for $P \in \mathcal{M}$. Then, the function-valued and null parameters in \eqref{nullhypothesis} are $\Psi_{0,x_{s}}=\E[Y^{(1)}-Y^{(0)}|X_{s}=x_{s}]$ and  $\theta_{0}= \E[Y^{(1)}-Y^{(0)}]$, respectively. In this case, the null hypothesis of interest,
\begin{align*}
    H_0:\E[Y^{(1)}-Y^{(0)}|X_{s}=x_{s}] =\E[Y^{(1)}-Y^{(0)}]  \quad P_{0,X_{s}}\text{-almost surely},
\end{align*}
also falls within the general framework of \eqref{nullhypothesis}. In the literature,  the parameters $\theta_{0}$ and $\Psi_{0,x_{s}}$ are referred to as the ATE and CATE, respectively. 

The two illustrative examples (Examples~\ref{examp1}--\ref{examp2}) represent many other statistical testing problems that can be motivated under the general form in \eqref{nullhypothesis}. We will present detailed discussion of the application of the proposed test to these examples in Section~\ref{illustration} and evaluate finite sample performance through simulations in Section~\ref{simulation}. 

\section{Proposed methodology\label{methodology}}

\subsection{Testing strategy \label{strategy}}
We next develop a test for hypotheses of the form  \eqref{nullhypothesis}. Suppose $ \Psi_{0,v} \in L^{p}(P_{0}) $ and $\theta_{0}$ is finite. Then, the null hypothesis in \eqref{nullhypothesis} is equivalent to  $\| \Psi_{0,v}-\theta_{0}\|_{L^{p}(P_{0})}=0$. When expressing the null hypothesis as a comparison between the norm and zero, one seemingly straightforward option is to perform a Wald-type test based on an $n^{1/2}$-consistent asymptotically normal estimator \citep{cai2024test,dai2022significance, williamson2021nonparametric}. However, such estimators do not admit to a tractable limiting null distribution due to the boundary null problem: under the null, the first-order influence function often vanishes, leading to a degenerate limiting null distribution of the test statistic \citep{verdinelli2024decorrelated}.
To address this challenge, we propose a test statistic motivated by an alternative representation of norms as supremum (via H\"older's inequality) over a suitable function class \citep[][Proposition~6.13]{folland1999real}. Specifically, the null hypothesis in \eqref{nullhypothesis} can be equivalently formulated as 
\begin{align}
   H^{\prime}_{0}: 
   \underset{\substack{h \in L^{q}(P_{0}) \\ \|h \|_{L^{q}(P_0)}=1}}{\sup} 
   \left | \int \left\{\Psi_{0,v}(v)-\theta_{0}\right\} h(v) \, dP_{0,V}(v) \right |=0, \label{newhypothesis}
\end{align}
where $q$ is the conjugate exponent of $p$ $(i.e., \frac{1}{p} +\frac{1}{q}=1)$ and $L^{q}(P_{0})$ is the dual of $ L^{p}(P_{0})$. We illustrate this formulation for $p \in \{1,2\}$. For $p=1$, the maximizer of supremum in \eqref{newhypothesis} is $h^{*}=\text{sign}\{\Psi_{0,v}-\theta_{0} \}$, yielding $\int \{\Psi_{0,v}(v)-\theta_{0}\} \text{sign}\{\Psi_{0,v}(v)-\theta_{0} \} dP_{0,V}(v) =\int |\Psi_{0,v}(v)-\theta_{0}|dP_{0,V}(v)$. For $p=2$, $h^{*}=(\Psi_{0,v}-\Psi_{0})/\|\Psi_{0,v}-\Psi_{0}\|_{L^2(P_0)}$, and the supremum becomes $\int \{\Psi_{0,v}(v)-\theta_{0}\} (\Psi_{0,v}(v)-\theta_{0})/\|\Psi_{0,v}-\theta_{0}\|_{L^2(P_0)} dP_{0,V}(v) = \left[\int (\Psi_{0,v}(v)-\theta_{0})^{2}dP_{0,V}(v)\right ]^{1/2}.$

Motivated by this dual representation, one could attempt to construct a test by estimating all linear transformations of $\Psi_{0,v} - \theta_0$ and explicitly evaluating the supremum over  $L^q(P_{0})$. However, this is technically challenging and practically infeasible, as $L^q(P_{0})$ is too large. Instead, we restrict attention to a rich but non-complex function class $\mathcal{H}$ (e.g., one with bounded entropy) that captures meaningful deviations from the null. Specifically, defining  
\begin{equation}\label{eqn:hdefs}
\bar{h}_{0}\coloneq  \int h(v)dP_{0,V}(v), \quad h^{c}_{P_{0,V}}\coloneq h(V) -\bar{h}_{0}, \quad \text{and} \quad \Omega_{P_{0}}(h)\coloneq \int \Psi_{0,v}(v)h^{c}_{P_{0,V}}(v)dP_{0,V}(v), 
\end{equation}
for $h \in \mathcal{H} \subset L^q(P_{0})$, we consider the following null hypothesis
 \begin{align}
     H_0^{\prime \prime}:  \underset{h \in \mathcal{H}}{\text{sup}} \left |\Omega_{P_{0}}(h) \right |=0. \label{hypo:3}
 \end{align}
From \eqref{hypo:3}, we construct our inference procedure by choosing $T_{n}(\mathcal{H})=\underset{h \in \mathcal{H}}{\text{sup}} \left |n^{1/2}\widehat{\Omega}_{n}(h) \right |$ as the test statistic, where $\widehat{\Omega}_{n}(h)$ is an EIF-based estimator of $\Omega_{P_{0}}(h)$. This is attractive because, as shown later in Theorem~\ref{theorem2}, $T_{n}(\mathcal{H})$ converges weakly to a continuous transformation of a Gaussian process under some regularity conditions and hence has a tractable limiting null distribution. 
Moreover, if $H_0$ is true, then $H_{0}^{\prime\prime}$ is also true, although the converse does not necessarily hold: $H^{\prime\prime}_{0}$ being true does not necessarily imply that $H_0$ is true if  $\mathcal{H}$ is not dense in $L^q(P_{0})$. Nevertheless, the test based on \eqref{hypo:3} ensures valid type-I error control for $H_0$ for any appropriate choice of $\mathcal{H}$, and it has power against alternatives aligned with the chosen class $\mathcal{H}$. The remainder of this section  describes the construction of the test in detail.

\subsection{Theoretical foundation \label{theoretical}}
Recall that for $P_{0} \in \mathcal{M}$ and a smooth map $\Psi$, we define $\Psi_0 \coloneqq \Psi(P_0)$ and $\Psi_{0,v} \coloneqq \Psi(P_{0,O|v})$. 
We take a path through $P_0$ in the direction $s$ as $dP_t=(1+ts)dP_0$, where $s$ is a function satisfying $\int s(o)dP_{0}(o)=0$ and $\int s^2(o)dP_{0}(o)<\infty$. Similarly, for the conditional distribution of $O$ given $V = v$, under $P_{0}$, we take a path through $P_{0,O|v}$ in the direction $s_{v}$ as $dP_{t,O|v}=(1+ts_{v})dP_{0,O|v}$, where $s_v$ is uniformly bounded and satisfies $\int s_{v}(o)dP_{0,O|v}(o)=0$. We simplify the notation by writing $P_0s \coloneq \int s(o)dP_{0}(o)$ and $\mathbb{P}_n s \coloneq \frac{1}{n}\sum_{i=1}^n s(O_i)$. To construct an EIF-based estimator of 
$
\Omega_{P_{0}}(h) = \int{\Psi_{0,v}(v) h^c_{P_{0,V}}(v) dP_{0,V}(v)},
$
%in the nonparametric model, 
defined in \eqref{eqn:hdefs}, 
we require $\Psi_{0,v}$ and $\Psi_{0}$ to satisfy the following regularity conditions:

\noindent{\bf {C1:} }\label {C1} For some uniformly bounded function  $D_{P_{0}}: \mathcal{O}  \rightarrow \mathbb{R}$ such that  $\int D_{P_{0}}(o)dP_{0}(o)=0$, 
    \begin{align*}   
   &\frac{d\Psi_{t}}{dt} \Bigg|_{t=0}= \int D_{P_{0}}(o)s(o)dP_{0}(o).
    \end{align*}
\noindent{\bf {C2:} } \label {C2}   For some uniformly bounded function  $D^{v}_{P_{0}} : \mathcal{O} \rightarrow \mathbb{R}$ such that  $\int D^{v}_{P_{0}}(o) dP_{0,O|v}(o)=0$, 
\begin{align*}
    & \frac{d \Psi_{t,v}}{dt}   \Bigg|_{t=0} =\int D^{v}_{P_{0}}(o)s_{v}(o)dP_{0,O|v}(o) \text{ holds $P_{0,V}$-almost surely}.
     \end{align*}
Condition \hyperref[C1]{\rm{C1}} ensures that  $\Psi(P)$ is pathwise differentiable at $P_{0}$ relative to $\mathcal{M}$ with EIF $D_{P_{0}}$ \citep{bickel1993efficient}. Condition \hyperref[C2]{\rm{C2}} is a conditional version of pathwise differentiability, in that, we take the G$\mathrm{\hat{a}}$teaux derivative of $\Psi_{t,v}$ with respect to the conditional distribution of $O$ given $V=v$ induced by $P_0$. We will refer to $D^{v}_{P_{0}}$ as the conditional influence function with respect to $P_{0,O|v}$.  Our next result characterizes the EIF of $\Omega_{P_{0}}(h)$. 
\begin{theorem}
\label{theorem1}
   Assume conditions \hyperref[C1]{\rm{C1}} and \hyperref[C2]{\rm{C2}} hold. Then, the functional $\Omega_{P_{0}}(h)$ is pathwise differentiable with efficient influence function
\begin{align*}
  D^*_{h,0}(o)= \Big(D^{v}_{P_{0}}(o) -\int\Psi_{0,v}(v)dP_{0,V}(v) +   \Psi_{0,v} (v)\Big)
h^{c}_{P_{0,V}}(v)- \Omega_{P_{0}}(h).
\end{align*} 
\end{theorem} 

Theorem~\ref{theorem1} implies that, under some regularity conditions, we can construct an estimator, $\widehat{\Omega}_n(h)$, of $\Omega_{P_{0}}(h)$ based on the EIF, $D^*_{h,0}$, that converges to a Gaussian distribution for a fixed $h$. Assuming further that $\mathcal{H}$ is $P_0$-Donsker, we can extend this convergence to be uniform over $\mathcal{H}$. That is, $\left\{n^{1/2}(\widehat{\Omega}_n(h)-\Omega_{P_{0}}(h)) : h \in \mathcal{H}\right\}$ will converge to a mean-zero Gaussian process $\mathbb{G}$. Therefore, defining our test statistic as some continuous function of  $n^{1/2}(\widehat{\Omega}_n(h))$, say $T_{n}(\mathcal{H})=\underset{h \in \mathcal{H}}{\text{sup}} \Big |n^{1/2}(\widehat{\Omega}_n(h)) \Big |$, guarantees a well-defined limiting null distribution for our test, as well as theoretical size guarantees for testing $H_{0}$.

\subsection{Calculating the conditional influence function}\label{sec:CEIF}
Obtaining $D^*_{h,0}$ from Theorem~\ref{theorem1} requires deriving the 
conditional influence function, $D^{v}_{P_{0}}$, associated with $\Psi_{0,v}$. Therefore, the derivation of $D^{*}_{h,0}$ is problem-dependent, often requiring knowledge of functional derivatives. In what follows, we show that  $D^{v}_{P_{0}}$ can be readily obtained for certain classes of function-valued  parameters. In many cases, $D^{v}_{P_{0}}$ can be obtained directly from the EIF of a known finite-dimensional parameter available in the semiparametric literature.

The first class of function-valued parameters we discuss consists of those for which the evaluation of the smooth map $\Psi(P)$ can be expressed as a continuous linear functional of the conditional mean function. To discuss this class, we consider a more general data structure $O=(U,V)$, where $U=(U_{1},U_{2})$. Here, $V$ denotes the variable indexing the function-valued parameter (e.g., a subgrouping variable), $U_1$ denotes a response variable of interest (e.g., a clinical outcome), and $U_2$ denotes additional covariates in the conditional mean function (e.g., a binary treatment indicator).

Denote $\mu_{0}(u_{2},v)=\E[U_{1} \mid U_{2}=u_{2},V=v]$ as the conditional mean function and $m(o,\mu_{0}(u_{2},v))$ be a continuous linear functional of $\mu_{0}(u_{2},v)$. Define the evaluation map as $\Psi(P)= \int m(o,\mu_{0}(u_{2},v)) \, dP(o)$, with $\Psi_{0}=\E[m(O,\mu_{0}(U_{2},V))]$ and $\Psi_{0,v}=\E[m(O,\mu_{0}(U_{2},V)) \mid V=v]$.
Following \cite{chernozhukov2018doubleb}, if $\E[m(O,\mu_{0}(U_{2},V))]=\E[\alpha_{0}(U_{2},V)\mu_{0}(U_{2},V)]$ for some function $\alpha_{0}$ satisfying $\E[\alpha_{0}(U_{2},V)^{2}]<\infty$, then the EIF of $\Psi_{0}$ takes the form
\begin{align}\label{eqn:prop1eif}
    D_{P_{0}}(o) = m(o,\mu_{0}(u_{2},v)) - \Psi_{0} + \alpha_{0}(u_{2},v)\{u_{1} - \mu_{0}(u_{2},v)\}.
\end{align}
Building on \eqref{eqn:prop1eif}, which holds for a broad class of well-studied statistical functionals, we show that when $\Psi_{0,v}$ is a continuous linear functional of the conditional mean, $D^{*}_{h,0}$ admits a closed-form expression.
\begin{proposition}\label{prop}
Let $\alpha_{0}: \mathcal{O} \rightarrow \mathbb{R}$ such that $\E[\alpha_{0}(U_{2},V)^{2}]<\infty$ and suppose
\begin{align*}
    \Psi_{0,v}
    = \E[m(O,\mu_{0}(U_{2},V))\mid V=v]
    = \E[\alpha_{0}(U_{2},V)\mu_{0}(U_{2},V) \mid V=v],
    %\quad \text{for all } \mu_{0} \text{ with } \E[\mu_{0}(U_{2},V)^{2}] < \infty.
\end{align*}
for all $\mu_{0}$ satisfying $\E[\mu_{0}(U_{2},V)^{2}] < \infty$.
Then, the EIF of $\Omega_{P_{0}}(h)$, $D^{*}_{h,0}(o)$, is given by
\begin{align*}
    \Big(
        m(o,\mu_{0}(u_{2},v))
        + \alpha_{0}(u_{2},v)\Big[u_{1} - \mu_{0}(u_{2},v)\Big]
        - \int \Psi_{0,v}(v)\, dP_{0,V}(v)
    \Big)
    h^{c}_{P_{0,V}}(v)
    - \Omega_{P_{0}}(h).
\end{align*}
\end{proposition} 
\noindent We show in Section~\ref{illustration} that Examples~\ref{examp1} and~\ref{examp2} from Section~\ref{examples} fall within the class of continuous linear functional of the conditional mean function; hence, in these examples, $D^{*}_{h,0}$ can be obtained directly using Proposition~\ref{prop}. The conditional covariance functional, discussed in the Supplementary Material~\ref{sec:additional-examples}, does not belong to this class. In what follows, we present another broad class of functionals that includes the conditional covariance functional. Consider the evaluation map, $\Psi(P)=\int g(o;P)dP(o)$, where $g$ depends on $P$ through derived marginal or conditional distributions under $P$. Proposition~\ref{prop2} shows that for function-valued parameters of the form $\Psi_{0,v}(v) = \E[g(O;P_{0,O|v})|V=v]$, the explicit form of $D^{*}_{h,0}(o)$ can be derived directly from the EIF of the corresponding marginalized parameter, $\int \E[g(O;P_{0,O|v})|V=v]dP_{0,V}(v)$. The following result formalizes this statement.

\begin{proposition}\label{prop2}
Define the smooth map $P \rightarrow \Psi(P)=\int g(o;P)dP(o)$. Consider the function-valued parameter $\Psi_{0,v}=\E[g(O;P_{0,O|v})|V=v]$ and its corresponding marginalized finite-dimensional parameter $\widetilde{\Psi}_{0}=\int \Psi_{0,v}dP_{0,V}(v)$. Assume $\widetilde{\Psi}_{0}$ is pathwise differentiable at $P_{0}$ relative to the nonparametric model $\mathcal{M}$ with EIF $\widetilde{D}_{P_{0}}$. If there exists $D^{v}_{P_{0}}$ such that Condition~\hyperref[C2]{\rm C2} holds, then the EIF of $\Omega_{P_{0}}(h)$, denoted $D^{*}_{h,0}(o)$, simplifies to
\begin{equation*}\label{eqn:prop2_2}
D^{*}_{h,0}(o) =
\widetilde{D}_{P_{0}}(o) 
h^{c}_{P_{0,V}}(v)
- \Omega_{P_{0}}(h).
\end{equation*}
\end{proposition}

\noindent Similar to Proposition~\ref{prop}, this result shows that the explicit form of $D^{*}_{h,0}$ can be derived directly from known EIFs, without requiring a separate derivation of the conditional influence function.  
The class of function-valued parameters covered by Proposition~\ref{prop2} is broad, encompassing the conditional mean, conditional variance, conditional expected density, and conditional covariance, among others. For instance, in the context of testing hypotheses about conditional covariance, the form of $\widetilde{D}_{P_{0}}$ is well established in the literature \citep{kennedy2024semiparametric, xiang2020flexible}.  
Together, Propositions~\ref{prop} and~\ref{prop2} provide a unified strategy for deriving $D^*_{h,0}$ across a wide range of problems by leveraging known EIFs of corresponding marginalized parameters.

\begin{remark} \label{remark1}
The conclusion from Proposition~\ref{prop} extends to a broader class of function-valued parameters beyond those that can be expressed as bounded linear functionals of the conditional mean. Specifically, this result can be extended to the class of functionals of the form $\Psi_{0}=\E[m(O,\eta_{0,v})]$ and $\Psi_{0,v}=\E[m(O,\eta_{0,v})|V=v]$, where $m(o,\eta_{0,v})$ is a smooth function of some summary $\eta_{0,v}$ of the conditional distribution $P_{0,O|v}$. If the conditional influence function for $\eta_{0,v}$ exist, then one can obtain the conditional influence function of $\Psi_{0,v}$ from the form of the efficient influence function of $\Psi_{0}$. If $\eta_{0,v}$ is the conditional mean, this result reduces to Proposition~\ref{prop}. However, $\eta_{0,v}$ can be any conditional quantity for which we can obtain its conditional influence function. An example is when $\eta_{0,v}$ is the conditional survival function. We present this general result in the Supplementary Material~\ref{sec:additional-examples} and derive the form of $D^{v}_{P_{0}}$ for when $\eta_{0,v}$ is the conditional survival function.
\end{remark}

\subsection{Test statistic and its limiting null distribution \label{estimation}}

We propose our test statistic as a function of an estimator of 
$\Omega_{P_{0}}(h)$. Importantly, we want the estimator to be flexible 
but also attain a parametric rate of convergence. To this end, we introduce the naive plug-in estimator, from which we build an estimator that achieves $n^{1/2}$-rate of convergence. Let $\widehat{P}_{n} \in \mathcal{M}$ be an estimate of $P_{0}$ based on the random sample $O_1, O_2, \ldots, O_n$, where $O_i = (U_i, V_i)$ for 
$i = 1, \ldots, n$. For $h \in \mathcal{H}$, we construct a plug-in estimator for the parameter $\Omega_{P_{0}}(h)$ in \
\eqref{eqn:hdefs} as 
    \begin{align*}
    \widehat{\Omega}_{\widehat{P}_{n}}(h) = \frac{1}{n}\sum_{i=1}^n  \Psi_{n,v}(V_{i}) \Big [ h(V_i)-\frac{1}{n}\sum_{j=1}^n h(V_j)\Big ],
    \end{align*}
where $\Psi_{n,v}$ is a nonparametric estimator of $\Psi_{0,v}$. Although the plug-in estimator is simple to obtain, it generally inherits bias from estimating nuisance functions with flexible nonparametric estimators. To correct the bias of the plug-in estimator, we propose to use the one-step bias correction method \citep{pfanzagl1985contributions} by adding a correction term based on the EIF of $\Omega_{P_{0}}(h)$. Under some regularity conditions, the de-biased estimator retains a normal asymptotic limit. We construct a one-step estimator, $\widehat{\Omega}^{os}_{n}(h)\coloneq \widehat{\Omega}_{\widehat{P}_{n}}(h) + n^{-1}\sum_{i=1}^n  D^*_{h,n}(O_i)$, at a fixed $h$ where 
\begin{align}
     D^*_{h,n}(O_i) & =   \frac{1}{n} \sum_{i=1}^n \Big[ D^{v}_{n}(O_i) -\sum_{k=1}^n \Psi_{n,v}(V_{k}) + \Psi_{n,v}(V_{i}) \Big]\Big (h(V_i)-\frac{1}{n}\sum_{j=1}^n h(V_j)\Big) -\widehat{\Omega}_{\widehat{P}_{n}}(h). \label{EIFn}
\end{align}
Here, $\Psi_{n,v}$ and $D^{v}_{n}$ could depend on other nuisance functions (e.g., the conditional mean function) that are specific to the function-valued parameter in the testing problem. 

In Theorem~\eqref{theorem2}, we state sufficient conditions under which $\{n^{1/2}\widehat{\Omega}^{os}_{n}(h): h \in \mathcal{H}\}$ has a tractable limiting distribution, thereby ensuring asymptotic validity of our proposed testing procedure.
In particular, we provide conditions under which $\widehat{\Omega}^{os}_{n}$ is uniformly asymptotically linear with influence function $D^*_{h,0}$ and has a Gaussian limit; that is, 
\begin{align}\label{eqn:omegahatos}
    \widehat{\Omega}^{os}_{n}(h)- \Omega_{0}(h)=\frac{1}{n}\sum_{i=1}^{n}D^*_{h,0}(O_{i})+r_n(h),
\end{align}
where $\underset{ h \in \mathcal{H}}{\text{sup }}|r_n(h)|=o_P(n^{-1/2})$, and the standardized process $\{n^{1/2}[\widehat{\Omega}^{os}_{n}(h) - \Omega_0(h)]: h \in \mathcal{H}\}$ is asymptotically Gaussian.

%\begin{samepage}
\begin{theorem} 
\label{theorem2}
Suppose $ \; i)  \; \mathcal{H}$ is a $P_{0}$-Donsker class, ii) \; $\{ D^*_{h,n}, D^*_{h,0}: h \in \mathcal{H}\}$ is contained in a $P_0$-Donsker class, and $(iii) \; \underset{ h \in \mathcal{H}}{\text{sup }}|r_n(h)|=o_P(n^{-1/2})$, where $r_n(h)$ is defined in \eqref{eqn:omegahatos}. Then, $\{n^{1/2}[\widehat{\Omega}^{os}_{n}(h)-\Omega_{0}(h)]: h\in \mathcal{H}\}$ converges weakly in the space $\ell^{\infty}(\mathcal{H})$ of bounded real-valued functions on $\mathcal{H}$ to a mean-zero Gaussian process $\mathbb{G}$ with covariance function $\Sigma_0(h_1,h_2)=P_0[D^*_{h_1,0}D^*_{h_2,0}]$.
\end{theorem} 
%\end{samepage}
%
Conditions $(i)$ and $(ii)$ restrict the complexity of the nuisance functions, their estimators, and the function class $\mathcal{H}$. The Donsker condition on $\{ D^*_{h,n}, D^*_{h,0}: h \in \mathcal{H}\}$ may be implied by the Donsker condition on \(\mathcal{H}\) under a Lipschitz-type condition. See Lemma S1 in the Supplementary Material of \cite{hudson2026inference} for details. Donsker classes are often quite large, so this assumption is not overly restrictive. Moreover, the Donsker condition is satisfied by a broad range of function classes, including functions with bounded variation norm, functions with finite VC dimension (e.g., finite-depth trees), and reproducing kernel Hilbert spaces. Nonparametric estimators can also be used, and many are known to satisfy these properties under weak conditions. Examples include penalized estimators and sieve-based procedures \citep{geer2000empirical}. For other function classes, the Donsker assumption can be verified by analyzing the metric entropy of the function class \citep{van2000asymptotic, kosorok2008introduction}. Finally, condition $(iii)$ pertains to the remainder term, ensuring that the estimator is uniformly asymptotically linear. This condition is often met by the one-step estimator, which we illustrate through specific examples in Section~\ref{illustration}.

\begin{algorithm}[t]
\caption{Bootstrap test procedure} \label{Algo1}
%\begin{algorithmic}[1]
    \textbf{Input:} Observed data $\{O_i\}_{i=1}^n$, $D^{*}_{h,n}(o)$, number of bootstrap samples $B$, significance level $\alpha$. \\
    \textbf{1: } \textbf{For each bootstrap sample} $b = 1, \dots, B$: \\
    
    \qquad \textbf{1a: } Draw $\xi^{(b)}_1, \dots, \xi^{(b)}_n$ 
   
    \qquad  \textbf{1b: } Compute $T^{(b)}(\mathcal{H}) = \sup_{h \in \mathcal{H}} \Big| \frac{1}{n} \sum_{i=1}^n \xi^{(b)}_i D^{*}_{h,n}(O_i) \Big|$\\
     
   \textbf{2: } Compute the observed test statistic $T_n(\mathcal{H}) = \sup_{h \in \mathcal{H}} \Big| \frac{1}{n} \sum_{i=1}^n D^{*}_{h,n}(O_i) \Big|$\\
   \textbf{3: } Compute the bootstrap $p$-value $\rho_B = \frac{1}{B} \sum_{b=1}^B \mathbbm{1}\Big\{ T^{(b)}(\mathcal{H}) > T_n(\mathcal{H}) \Big\}$\\
   \textbf{4: } \textbf{Reject} $H_0$ at level $\alpha$ if $\rho_B \le \alpha$
%\end{algorithmic}
\end{algorithm}

Leveraging Theorem~\ref{theorem2}, we choose as our test statistic $T_{n}(\mathcal{H})\coloneq\underset{h \in \mathcal{H}}{\text{sup }}n^{1/2}|\widehat{\Omega}^{os}_{n}(h)|$. By the continuous mapping theorem,  $\underset{h \in \mathcal{H}}{\text{sup }}n^{1/2}|\widehat{\Omega}^{os}_{n}(h)|$ converges weakly to $\underset{h \in \mathcal{H}}{\text{sup }}|\mathbb{G}(h)|$ under the null. Therefore, a test that rejects the null when $T_n(\mathcal{H})$ exceeds the $1-\alpha$ quantile of $\underset{h \in \mathcal{H}}{\text{sup }}|\mathbb{G}(h)|$ achieves type-I control at the level $\alpha$. While integrating under the exact distribution of $\underset{h \in \mathcal{H}}{\text{sup }}|\mathbb{G}(h)|$ is difficult, we can use an approximation instead. To approximate the limiting process $\underset{h \in \mathcal{H}}{\text{sup }}|\mathbb{G}(h)|$, we draw $B$ sets of $n$ $i.i.d.$ samples from a distribution with mean zero and unit variance, denoted $\xi^{(b)}_{1},\xi^{(b)}_{2},\ldots,\xi^{(b)}_{n}$ for $b=1,\cdots, B$.  
% the set of $n$ random samples from a mean zero and unit variance distribution. We calculate 
Then, the multiplier bootstrap statistic, $T^{(b)}(\mathcal{H}) = \underset{h \in \mathcal{H}}{\text{sup }}|n^{-1/2}\sum_{i=1}^n\xi_{i}D^{*}_{h,n}(O_{i})|$, 
% The so-called multiplier bootstrap samples $T^{(1)}(\mathcal{H}), T^{(2)}(\mathcal{H}),\ldots,T^{(m)}(\mathcal{H})$ 
approximate the distribution of $\underset{h \in \mathcal{H}}{\text{sup }}|\mathbb{G}(h)|$ for large $n$ and $B$; 
see \citet{hsu2016multiplier} and \citet{hudson2026inference} for discussions of such simulated processes. The detailed procedure for constructing the test for null hypotheses \eqref{hypo:3} is presented in Algorithm \ref{Algo1}.

%\textcolor{red}{Local power?}

\section{ Illustrations \label{illustration}}

In this section, we show how our proposed framework applies to Examples~\ref{examp1}--\ref{examp2} introduced in Section~\ref{examples}. Specifically, for each example, we present the EIF, $D^{*}_{h,0}$, and the one-step estimator, $\widehat{\Omega}^{os}_{n}(h)$ of $\Omega_{P_0}(h)$. We defer discussion of conditions under which $\underset{h \in \mathcal{H}}{\text{sup }} |r_n(h)|=o_{p}(n^{-1/2})$ and detailed derivation of the presented results to the Supplementary Material~\ref{sec:details-examples}.

\subsection*{ Example 1 (continued) }
    Recall that the null hypothesis for testing conditional mean independence is 
    \begin{align*}
        H_{0}:  \E[Y|X=x]=\E[Y] \quad P_{0,X}\text{-almost surely}.
    \end{align*}
    Here, $\theta_{0}=\E[Y]$, $\Psi_{0,x}=\E[Y|X=x]$, and $\Omega_{P_{0}}(h)=\int \E[Y|X=x]h^{c}_{P_{0,X}}(x)dP_{0,X}(x)$. The EIF for $\Omega_{P_0}(h)$ can be characterized using Proposition~\ref{prop}. For $\mu_0(x) = E[Y|X = x]$, $\alpha_{0}(o)=1$ and $m(o,\mu_{0}(x))=\mu_0(x)$, we can express $\Psi_{0,x}$ as a continuous linear functional of the conditional mean function. That is, $\Psi_{0,x}=\E[ m(o,\mu_{0}(X))|X=x]$. Hence, following Proposition~\ref{prop}, $D^{x}_{P_{0}}=y-\E[Y|X=x]$, and the EIF, $D^{*}_{h,0}$, and one-step estimator, $\widehat{\Omega}^{os}_{n}(h)$, of $\Omega_{P_{0}}(h)$ are given by
    \begin{align*}
          D^*_{h,0}(o) &= (y -\E[Y])
h^{c}_{P_{0,X}}(x)-  \Omega_{P_{0}}(h), \\
 \widehat{\Omega}^{os}_{n}(h)&= n^{-1} \sum_{i=1}^n ( Y_{i} - \bar{Y}_{n})h^{c}_{n}(X_{i}),
    \end{align*}
    where $h^{c}_{n}(X)= h(X)-\bar{h}_{n}$, and
    $\bar{Y}_{n}=\frac{1}{n}\sum_{i=1}^{n}Y_{i}$, and $\bar{h}_{n}=\frac{1}{n}\sum_{i}^{n}h(X_{i})$ are estimators of $\E[Y]$, and $\E[h(X)]$, respectively.

\subsection*{ Example 2 (continued) }
The null hypothesis for testing treatment effect heterogeneity is 
\begin{align*}
    H_0:\E[Y^{(1)}-Y^{(0)}|X_{s}=x_{s}] =\E[Y^{(1)}-Y^{(0)}] \quad P_{0,X_{s}}\text{-almost surely}.
\end{align*}
Under standard causal assumptions of consistency, positivity, and conditional exchangeability \citep{hernan2020whatif}, the CATE and ATE can be identified in both randomized and observational studies as $\E[\E[Y|T=1,X]-\E[Y|T=0,X]|X_{s}]]$ and $\E[\E[Y|T=1,X]-\E[Y|T=0,X]]$, respectively.  Letting $\mu_{0}(t,x)=\E[Y|T=t,X=x]$ and $\pi_{0}(t|x)=\text{Pr}(T=t|X=x)$, we denote $\tau_{0,x_{s}}=\E[\mu_{0}(1,X)-\mu_{0}(0,X)|X_{s}=x_{s}]$ and $\tau_{0}=\E[\mu_{0}(1,X)-\mu_{0}(0,X)]$ as the identified CATE and ATE, respectively. Here, $\Psi_{0,x_{s}}=\E[\mu_{0}(1,X)-\mu_{0}(0,X)|X_{s}=x_{s}]$, $\theta_{0}=\E[\mu_{0}(1,X)-\mu_{0}(0,X)]$, and $\Omega_{P_{0}}(h)=\int \E[\mu_{0}(1,X)-\mu_{0}(0,X)|X_{s}=x_{s}] h_{P_{0,X_{s}}}^{c}(x_{s}) \}dP_{0,X_{s}}(x_{s})$. Similar to Example~\ref{examp1}, we can express $\Psi_{0,x_{s}}$ as a continuous linear functional of the conditional mean function. To see this, let $\alpha_{0}(t,x)=\frac{t}{\pi_{0}(t|x)}- \frac{1-t}{1-\pi_{0}(t|x)}$ and $m(x,\mu_{0})=\mu_{0}(1,x)-\mu_{0}(0,x)$. Then, we have $\Psi_{0,x_{s}}=\E [\alpha_{0}(T,X) \mu_{0}(T,X)|X_{s}=x_{s}]$.  Hence, following Proposition~\ref{prop}, the EIF, $D^{*}_{h,0}$, and one-step estimator, $\widehat{\Omega}^{os}_{n}(h)$, of $\Omega_{P_{0}}(h)$ are given by
\begin{align*}
          D^*_{h,0}(o) &= (\mu_{0}(1,x)-\mu_{0}(0,x)  +\alpha_{0}(t,x)\{y-\mu_{0}(t,x) \} - \tau_{0} )
h^{c}_{P_{0,X_{s}}}(x_{s})-  \Omega_{P_{0}}(h), \\
 \widehat{\Omega}^{\mathrm{os}}_{n}(h)&= n^{-1} \sum_{i=1}^n (\mu_{n}(1,X_{i})-\mu_{n}(0,X_{i})  +\alpha_{n}(X_{i})\{Y_{i}-\mu_{n}(T_{i},X_{i}) \} - \tau_{n} )\left (h(X_{s,i})-\bar{h}_{n}\right),
    \end{align*}
where we have replaced the subscript $``0"$ with $``n"$ to denote an estimator of the respective parameter (i.e., $\mu_{n}(t,x)$ is the estimator of $\mu_{0}(t,x)$). 

\section{An aggregate test for improved power}\label{aggregatetest}
So far, we have discussed our inferential strategy assuming knowledge of the function class $\mathcal{H}$.
In practice, no single choice of function class can guarantee good power across all alternatives. In fact, the optimal choice of function class for ensuring good power against specific alternatives is typically unknown. To ensure reasonable power against a broader class of alternatives, we propose constructing an aggregate test that combines test statistics across different pre-specified function classes. 

Let $L \in \mathbb{N}$ be a fixed natural number, and $\{\mathcal{H}_{l}\}_{l=1}^{L}$ denote a collection of function classes. For each $l$, let $T_{n}(\mathcal{H}_{l})$ be our test statistic constructed using function class $\mathcal{H}_{l}$, and define the $L$-dimensional vector of observed test statistics $\bm{T}_{n} \coloneqq (T_{n}(\mathcal{H}_{1}), \ldots, T_{n}(\mathcal{H}_{L})) \in \mathbb{R}^{L}$. To simplify notation, we denote $\bm{T}_{n}$ by $\bm{T}^{(0)}$. The idea of the aggregate test is based on assessing how extreme a scalar aggregate summary of $\bm{T}^{(0)}$ is relative to samples from a bootstrap approximation of the summary's sampling distribution. Let $\bm{T}^{(b)} = \{T^{(b)}(\mathcal{H}_{l})\}_{l=1}^{L}$, for $b = 1, \ldots, B$, denote the vector of size $L$ drawn from the approximate limiting null distribution of $\bm{T}^{(0)}$. Since $T^{(b)}(\mathcal{H}_{l})$ is on a different scale for each $l$, we standardize both the $b$th bootstrap vector and the observed statistic $\bm{T}^{(0)}$ by shifting and scaling using the mean and standard deviation computed from $\bm{T}^{(-b)}$, the collection of all vectors except $\bm{T}^{(b)}$ with $\bm{T}^{(0)}$ included. We include $\bm{T}^{(0)}$ in $\bm{T}^{(-b)}$ to ensure exchangeability of the observed and bootstrap statistics under the null. Algorithm~\ref{Algo2} presents detailed steps for performing an aggregate test when $\bm{T}^{(0)}$ is summarized using $\|\cdot\|^{2}_{L^{2}(\mathbb{P}_{n})}$. A similar strategy is presented in \citet{shah2018goodness}, which combines multiple test statistics obtained from multiple tuning parameters in a penalized high-dimensional regression setting. Alternative aggregating functions may also be used, including the absolute maximum and $\|\cdot\|_{L^{1}(\mathbb{P}_{n})}$.

\begin{algorithm}[t]
\caption{Aggregate test procedure}\label{Algo2}
%\begin{algorithmic}[1]
\textbf{Input:} Observed vector of test statistics $\bm{T}^{(0)}=\{T_{n}(\mathcal{H}_{l})\}_{l=1}^L$; $B$ bootstrap samples $\{\bm{T}^{(b)}\}_{b=1,\dots,B}$; significance level $\alpha$. \\
\textbf{1:} \textbf{For each bootstrap sample} $b = 1, \dots, B$: and observed test statistic $\bm{T}^{(0)}$, calculate $\bm{\hat{\mu}}^{(-b)}=\Big(\hat{\mu}_{1}^{(-b)}, \ldots, \hat{\mu}_{L}^{(-b)} \Big)$ and $\bm \hat{\sigma}^{(-b)} =\Big(\hat{\sigma}_{1}^{(-b)}, \ldots, \hat{\sigma}_{L}^{(-b)} \Big)$ as 

\[ 
\hat{\mu}_{l}^{(-b)} = \frac{1}{B} \sum_{\substack{b'=0\\ b'\neq b}}^{B} T^{(b')}(\mathcal H_l),\quad \hat{\sigma}_{l}^{(-b)} = \Big[ \frac{1}{B-1} \sum_{\substack{b'=0\\ b'\neq b}}^{B} (T^{(b')}(\mathcal H_l)-\hat{\mu}_{l}^{(-b)})^2 \Big]^{1/2}.
\] 
\qquad \textbf{1a:} Compute the observed aggregate test statistic 
\[ 
\tilde{Q}_0 = \frac{1}{L} \sum_{l=1}^L \Big( \frac{T_{n}(\mathcal{H}_{l}) - \hat{\mu}_{l}}{\hat{\sigma}_{l}} \Big)^2, 
\] 
where $\hat{\mu}_{l}$ and $\hat{\sigma}_{l}$ correspond to $\hat{\mu}_{l}^{(-b)}$ and $\hat{\sigma}_{l}^{(-b)}$ for when $b=0$. \\
\textbf{2:} \textbf{For each bootstrap sample} $b = 1, \dots, B$, 

\qquad \textbf{2a:} Compute the $b^{\text{th}}$ bootstrap transformed value using the aggregate function 
\[ \tilde{Q}_{b} = \frac{1}{L} \sum_{l=1}^L \Big( \frac{T^{(b)}(\mathcal{H}_{l}) - \hat{\mu}_{l}^{(-b)}}{\hat{\sigma}_{l}^{(-b)}} \Big)^2.
\] 

\textbf{3:} Calculate aggregate $p$-value; $P_{\text{aggreg}} = \frac{1}{B+1} \Big(1 + \sum_{b=1}^B \mathbb{I}\left\{ \tilde{Q}_b \ge \tilde{Q}_0 \right\} \Big).$ 

\textbf{4:} \textbf{Reject} $H_0$ at level $\alpha$ if $P_{\text{aggreg}} \le \alpha.$ 
%\end{algorithmic} 
\end{algorithm}

Under the null, the collection $\{\bm{T}^{(0)}, \bm{T}^{(1)}, \ldots, \bm{T}^{(B)}\}$ is exchangeable. Since $\|\cdot\|^{2}_{L^{2}(\mathbb{P}_{n})}$ is a symmetric aggregation function that preserves exchangeability, standard results for Monte Carlo tests \citep{davison1997bootstrap, shah2018goodness} imply that the $p$-value from the proposed aggregate test is valid, i.e., $\Pr(P_{\mathrm{aggreg}} \le \alpha \mid H_{0}) \le \alpha$. Other order-preserving continuous aggregation functions, such as  arithmetic mean or  maximum, could also be used instead of $\|\cdot\|^{2}_{L^{2}(\mathbb{P}_{n})}$. In the next section, we discuss specific examples of $\mathcal{H}$ and their merits and limitations.

\section{Choices of function class}
To obtain a test with desirable power properties,  $\mathcal{H}$ needs to be large enough such that whenever $H_{0}$ fails, there exist some $h^{*} \in \mathcal{H}$ with $|\Omega_{P_{0}}(h^{*})| > 0$.  We discuss two choices of $\mathcal{H}$ that are popular for their computational and statistical properties, namely, the thresholded indicator function class and the reproducing kernel Hilbert space. Other choices of function class can be seamlessly incorporated into the proposed aggregate testing framework. 

\subsection{Thresholded indicator}

First, recall that for $p=1$, the maximizer of \eqref{newhypothesis} among the class of bounded functions is $h^{*}=\text{sign}\{ \Psi_{0,v}\}$ up to constant proportionality. A simple function class that can approximate bounded functions on  $\mathcal{V}$ is the class of thresholded indicators defined on each value in $\mathcal{V}$. In fact, if $\Psi_{0,v}$ is monotone on $\mathcal{V}$, the maximizer, $h^{*}$, belongs to this class.  We define the class of thresholded indicators on $\mathcal{V}$ as $\mathcal{H}=\{ h_{v_0}(v)={\mathbbm{1}\{v \le v_0 \}} :  v_0 \in \mathcal{V} \}$. This choice of $\mathcal{H}$ is $P_{0}$-Donsker \citep{van2000asymptotic}, and the corresponding  test statistic is
 \begin{align*}
     T_n(\mathcal{H})=\underset{v_0 \in \mathcal{V}}{\text{max}} \quad\Big |\frac{1}{n} \sum_{i=1}^n\Big [ D_{n}^{v}(O_i) - \frac{1}{n}\sum_{j=1}^{n} \Psi_{n,v}(V_{j}) +\Psi_{n,v}(V_{i})\Big] \Big \{\mathbbm{1}\{V_i \le v_0 \}- \frac{1}{n}\sum_{k=1}^{n}\mathbbm{1}\{V_{k}\le v_0 \}\Big\}\Big |,
 \end{align*}
where, as before, $\Psi_{n,v}$ and $D^{v}_{n}$ denote the function-valued parameter and its corresponding conditional influence function with plug-in nuisance estimators.

This class is appealing due to its computational simplicity and ease of implementation. However, when $\Psi_{0,v}$ is highly complex (e.g., sinusoidal or quadratic), choosing $\mathcal{H}$ as simple as the class of thresholded indicators may result in a loss of power of the test. 

\subsection{Reproducing kernel Hilbert space} \label{RKHS_class}

The reproducing kernel Hilbert space (RKHS) is a rich function class that provides a convenient way to approximate the norm in \eqref{newhypothesis} when $p=2$. Thus, by selecting $\mathcal{H}$ as the RKHS,  we can obtain a powerful test even when $\Psi_{0,v}$ is non-monotone. We consider the RKHS induced by the eigenfunctions $\psi_{2j-1}(x) = \sqrt{2}\cos(2j\pi x)$ and $\psi_{2j}(x) = \sqrt{2}\sin(2j\pi x)$ with corresponding eigenvalues $\lambda_{2j-1}=\lambda_{2j}=(2j\pi)^{-4}$ for $j=1,2,\ldots$. Each function, $h$, in the space can be expressed as $h(x)=\sum_{j=1}^{\infty} a_j \psi_j(x)$ such that its norm satisfies $\sum_{j=1}^{\infty}\frac{a_j^{2}}{\lambda_{j}}<\infty$ for $a_{j}\in \mathbb{R}$. Because the eigenvalues decay as $\lambda_j \asymp j^{-4}$, the coefficients $a_j$ must decay sufficiently fast for the norm to remain finite. Consequently, the behavior of functions in $\mathcal{H}$ is largely determined by the low-frequency components of the expansion, and the contribution of higher-order terms becomes increasingly negligible. This motivates approximating $h$ using only the first $D$ basis functions for a sufficiently large integer $D$, which provides a good approximation to functions in $\mathcal{H}$. Let $\Gamma$ denote a $D\times D$ diagonal matrix with entries $\Gamma_{j,j} = 1/\lambda_j$, and let $\mathbf{V}$ denote a $D\times D$ matrix with entries given by estimates of the covariance between each pair of eigenfunctions. The RKHS norm and variance for a function $h(x)=\sum_{j=1}^{D} a_j \psi_j(x)$ are then given as $\bm{a}^{T} \Gamma \bm{a}$ and $\bm{a}^{T}\mathbf{V}\bm{a}$, respectively. To balance smoothness and variability, we consider the following constrained function class 
\begin{align*}
\mathcal{H} = \Big \{\sum_{j=1}^D a_j\psi_j(x) : a_j \in \mathbb{R}, \; \gamma \bm{a}^T\Gamma \bm{a} + (1-\gamma)\bm{a}^T\mathbf{V} \bm{a} \le C, \text{ for some fixed } C>0, \text{ and } \gamma\in [0,1] \Big \},
\end{align*}
where $\bm{a}=(a_1,a_2,\ldots,a_D)$. The parameter $\gamma$ controls the trade-off between smoothness and variability: values of $\gamma$ close to one place greater weight on the RKHS norm, $\bm{a}^T\Gamma\bm{a}$, and therefore enforce smoother functions; values close to zero place greater weight on the variance term,  $\bm{a}^T\mathbf{V}\bm{a}$, and therefore restrict the variability of the function; the constant $C$ determines the scale of the constraint and does not affect the shape of the function class. See Supplementary Material~\ref{sec:rkhs} for a more detailed discussion on the RKHS and motivation for the constraint on the function class.

%\subsubsection{Calculating the test statistic under RKHS}
A key advantage of using the RKHS is that we can explicitly express the one-step estimator, $\widehat{\Omega}_{n}^{\mathrm{os}}(h)$, of $\Omega_{P_{0}}(h)$ as a linear function in $\bm{a}$. That is,
\begin{align*}
    \widehat{\Omega}_{n}^{\mathrm{os}}(h) 
    &= \frac{1}{n} \sum_{i=1}^n \Big[ D^{v}_{n}(O_i) - \frac{1}{n}\sum_{k=1}^{n}\Psi_{n,v}(V_{k}) + \Psi_{n,v}(V_{i}) \Big ]
    \Big (\sum_{j=1}^D a_j\psi_j(V_i)-\frac{1}{n}\sum_{k=1}^n \sum_{j=1}^D a_j\psi_j(V_k)\Big) 
    = \bm{U}_n^{T}\bm{a},
\end{align*}
where $\bm{U}_n=\bm{U}(\bm{\phi}^*_n,\bm{\psi})=n^{-1}\bm{\Phi}^{T}\bm{\phi}^*_n$ is a $D$-dimensional vector. Here $\bm{\phi}^*_n$ is an $n$-dimensional vector whose $i$th entry is $D^{v}_{n}(O_i) - \frac{1}{n} \sum_{k=1}^{n}\Psi_{n,v}(V_{k}) + \Psi_{n,v}(V_{i})$, and $\bm{\Phi}$ is an $n\times D$ matrix whose $j$th column has entries $\psi_j(V_i)-\frac{1}{n}\sum_{k=1}^n\psi_j(V_k)$ for $i=1,2,\ldots,n$. Given the simplification above, we calculate our test statistic by solving the following maximization problem:
\begin{equation}
   T_n(\mathcal{H})= \underset{\bm{a} \in \mathbb{R}^D}{\max} \left \{n^{1/2} \bm{U}_n^T\bm{a} :   \gamma \bm{a}^T\Gamma \bm{a} + (1-\gamma)\bm{a}^T\mathbf{V} \bm{a}  \le C \right \}. 
   \label{optim1}
\end{equation}
The solution to problem \eqref{optim1} is obtained by maximizing the Lagrangian associated with \eqref{optim1}:
\begin{equation}
  \hat{\bm{a}}_{\eta,\gamma} =\underset{\mathbf{a} \in \mathbb{R}^D}{\text{argmax}} \left \{ n^{1/2}\bm{U}_n^T\bm{a}  -\frac{\eta}{2} \left(\gamma \bm{a}^T\Gamma \bm{a} + (1-\gamma)\bm{a}^T\mathbf{V} \bm{a} \right) \right \},
   \label{optim2}
\end{equation}
where $\eta>0$ is chosen to satisfy the constraint in \eqref{optim1}. 
For fixed $\eta$, the solution to \eqref{optim2} is given by
\begin{align*}
    \hat{\bm{a}}_{\eta,\gamma}
    = n^{1/2}\eta^{-1}
    \left(\gamma\Gamma+(1-\gamma)\mathbf{V}\right)^{-1}
    \bm{U}_n .
\end{align*}
We can thus express our test statistic as
\begin{align*}
    T_n(\mathcal{H})
    =
    \eta^{-1}
    n^{1/2}\bm{U}_n^T
    \left(\gamma\Gamma+(1-\gamma)\mathbf{V}\right)^{-1}
    n^{1/2}\bm{U}_n \, .
\end{align*}

The Lagrange multiplier $\eta$ primarily affects the scale of the test statistic but does not influence the shape of its sampling distribution. For this reason, we typically fix $\eta$ at a positive constant. In contrast, the tuning parameter $\gamma$ governs the trade-off between smoothness and variance of the estimated functions and directly influences the sampling distribution of the test statistic. Consequently, each value of $\gamma$ indexes a distinct function class and yields a corresponding test statistic with its own sampling distribution.

\section{Numerical studies \label{simulation}}
In this section, we investigate the finite sample performance of the proposed test in the context of the two working examples introduced in Section~\ref{examples}. For each example, we study various settings under the null and alternative hypotheses. We implement our test using the thresholded indicator and  RKHS as function classes. We truncate the basis expansion at $D = 100$ for the RKHS function class. From the RKHS, we construct $K=50$ function classes based on a grid of $50$ different $\gamma$ values, given by $\gamma_{max}\left[ \frac{\gamma_{min}}{\gamma_{max}}\right]^{\frac{k-1}{(K-1)}}$, where $k=1,2,\ldots,K$, $\gamma_{max}=1\times 10^{-3}$, and  $\gamma_{min}=1\times10^{-5}$. We implement the aggregate test by combining the thresholded indicator and the $K$ RKHS function classes. We approximate the null limiting distribution using $B = 800$ multiplier bootstrap samples.

As an alternative to the proposed aggregate test, we also consider the Cauchy combination method for combining multiple tests \citep{liu2020cauchy}, which we refer to as the Cauchy test. To investigate the advantages of combining different function classes, we also implement our procedure using individual function classes only, such as the thresholded indicators and the RKHS with only a smoothing penalty (corresponding to $\gamma = 1$), as well as an aggregate test combining only the $50$ RKHS function classes but without including the threshold class. We refer to the test based on the thresholded indicators as Indicator test, the RKHS with only a smoothing penalty as FixedRKHS test, and the test that combines the $50$ RKHS function classes as CombinedRKHS test. For Examples~\ref{examp1} and~\ref{examp2}, we compare the performance of the proposed tests with tests specifically designed for these problems. We conduct 500 Monte Carlo simulations with sample sizes $n \in \{125, 250, 500, 1000, 2000\}$ across all settings.

\subsection*{Example 1: \normalfont{Testing for conditional mean dependence}} 
For $X \sim \text{Unif}(-1,1)$ and $\varepsilon \sim N(0,1)$ independent of $X$, we consider three conditional outcome distributions:
\begin{itemize}
    \item Setting 1: $Y=\varepsilon$
    \item Setting 2: $Y=0.25X + \varepsilon$
    \item Setting 3: $Y=\text{sin}(\pi X \text{sign}(X) )+ \varepsilon$ 
\end{itemize}

\noindent For each setting, we compare the performance of our proposed tests with the approach by \citet{williamson2023general} implemented in the R package \texttt{vimp} \citep{vimp}. This approach assesses whether the expected least squared error loss resulting from predicting $Y$ using the conditional mean given $X$, falls below the marginal variance of $Y$. 

\begin{figure}[t]
    \centering
    \includegraphics[width=0.90\textwidth]{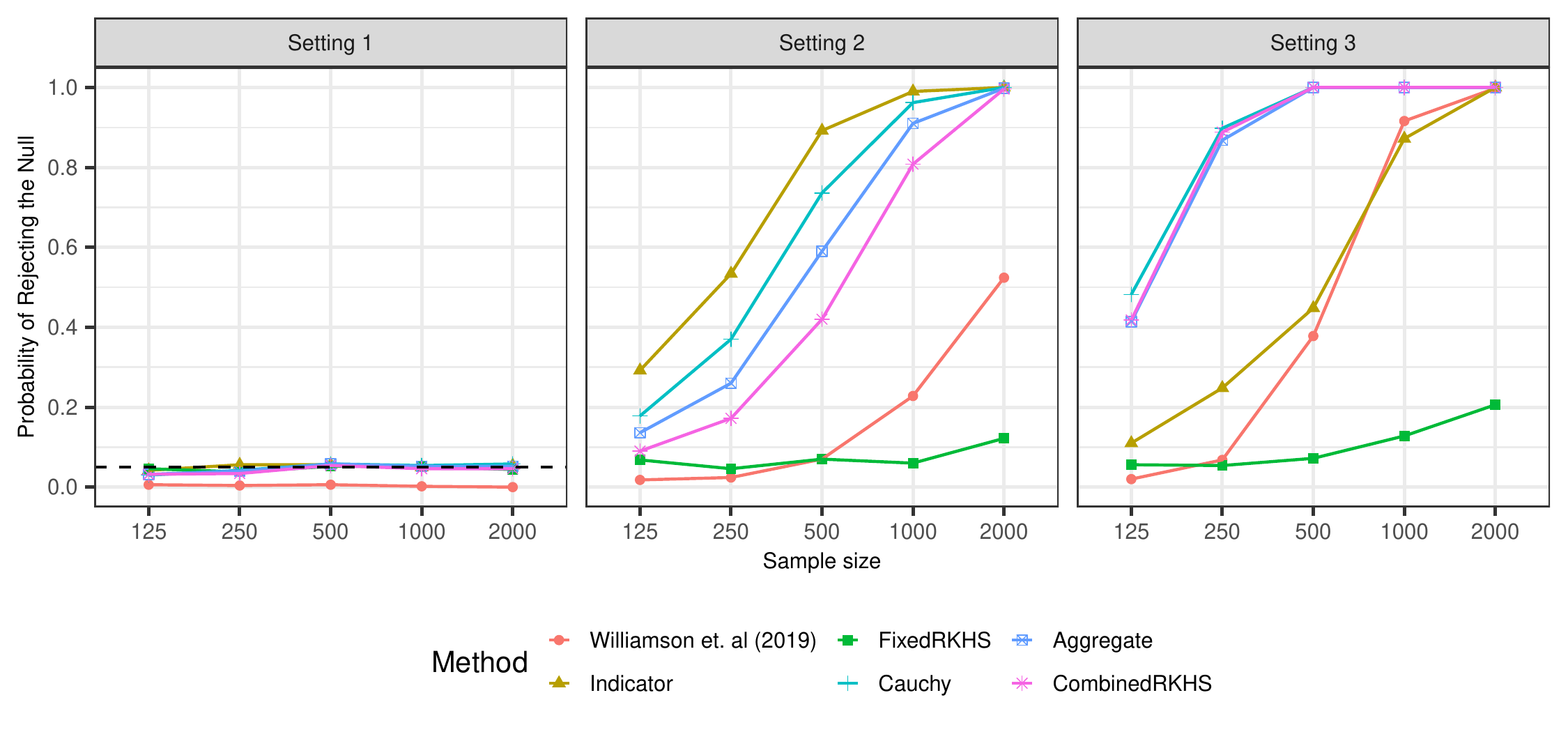}
    \caption{Empirical probability of rejection for hypotheses in Example 1 under different data generating mechanisms, at different sample sizes. Horizontal wide-dash line indicates the nominal 0.05 test size.}
\label{fig1}
\end{figure}

The results are shown in Figure~\ref{fig1}. Under Setting 1, where the null hypothesis $\E[Y|X]=\E[Y]$ holds, we observe that all the tests from our proposed framework control the type-I error. However, the method of \cite{williamson2023general} appears to be conservative. Under Setting 2, where the alternative is linear, and hence monotone, the Indicator test has the highest power. We observe slower increase in power as $n$ increases for the Cauchy and Aggregate tests since they are based on RKHS, which accommodates more complex functions. Nonetheless, the Aggregate and Cauchy tests are still able to achieve high power similar to the Indicator test for larger sample sizes. In Setting 3, where the conditional mean function is more complex compared to Setting~2, the Aggregate and Cauchy tests achieve high power at small sample size. The power of the Indicator test, on the other hand, increases more slowly with sample size. The power of the test of \cite{williamson2023general} increases with sample size in both Settings 2 and 3, but at a much slower rate than the other proposed tests. For the test based on $\gamma = 1$ under the RKHS class (``FixedRKHS''), we observe good type-I error control but very low power in both Settings~2 and 3. This occurs because the FixedRKHS test uses a more restrictive function class, which limits the ability of the test to detect departures from the null.

\subsection*{Example~2: \normalfont{Assessing treatment effect heterogeneity}}
We generate our treatment assignment variable $T$ randomly from $\{1,0\}$ with probability 0.5 corresponding to a randomized experiment, our conditioning covariate $W \sim \text{Unif}(-1,1)$, and white noise $\varepsilon \sim N(0,0.5)$ independent of $W$. We consider three settings for the outcome distribution, corresponding to different specifications of the CATE:
\begin{figure}[t]
    \centering
    \includegraphics[width=0.90\textwidth]{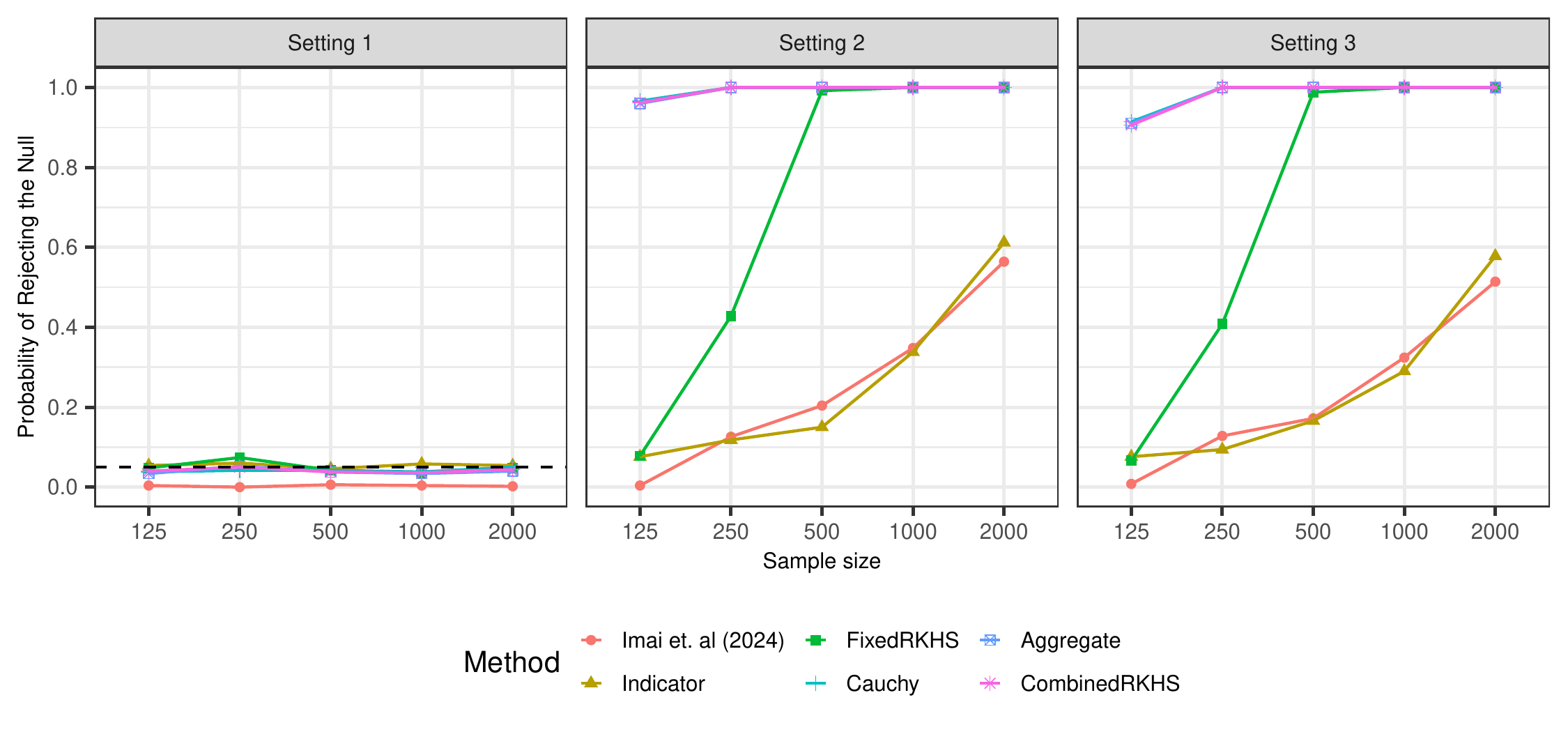}
    \caption{Empirical probability of rejection for hypotheses in Example 3 under different data generating
mechanisms, at different sample sizes. Horizontal wide-dash line indicates the nominal 0.05 test size.}
\label{fig3}
\end{figure}

\begin{itemize}
    \item Setting 1: $Y=\beta_0 + \beta_1W + \gamma T + \varepsilon$ 
    \item Setting 2: $Y=\beta_0 + \beta_1W + (\gamma_0 +\gamma_1 W ) T + \varepsilon$ 
    \item Setting 3:  $Y=\beta_0 + \beta_1W + (\gamma_0 +\gamma_1 \text{sin}[W] ) T + \varepsilon$ 
\end{itemize}
Under Setting 1, there is no  treatment effect heterogeneity. Settings 2 and 3 correspond to  linear and nonlinear effect heterogeneity, respectively. We estimate the nuisance functions $\pi_{0}(t|x)$ and $\mu_{0}(t,x)$ using generalized additive models \citep{hastie1990generalized} with natural cubic spline components, using the  \texttt{mgcv} R package. For each setting, we compare our proposed tests with the recently proposed test of \cite{imai2024statistical}, implemented in the \texttt{evalITR}  R package \citep{evalITR}. This method assesses treatment effect heterogeneity by testing equality of the sorted group average treatment effects (GATES) proposed by \cite{chernozhukov2017fisher}.
 All tests exhibit good control of type-I error, even when the sample size is small. The method of  \cite{imai2024statistical} is conservative as the sample size increases in Setting 1. Under Settings 2 and 3, the Aggregate and Cauchy tests achieve high power even for $n= 125$. In these settings the low power of the test by \cite{imai2024statistical} may be due to the sample splitting requirement of the test.

\section{Treatment effect heterogeneity in HER2+ BCa\label{application}} HER2-positive breast cancer (HER2-positive BCa) is a subtype driven by over-expression of the human epidermal growth factor receptor 2 (HER2) protein, often resulting from an excess of HER2 gene copies in cancer cells. This over-expression accelerates cell growth and progression, leading to a more aggressive disease course without treatment. According to the American Cancer Society, approximately $15-20\%$ of women with breast cancer are diagnosed with HER2+ disease. The standard treatment for HER2+ patients includes targeted therapies like trastuzumab, an intravenous antibody that inhibits HER2 production and has been shown to extend disease-free survival in affected women \citep{joensuu2006adjuvant}. 

Additional findings by \cite{prat2014based}, based on the NeoAdjuvant Herceptin (NOAH) clinical trial, suggest that HER2+ patients classified into the HER2-enriched (HER2-E) subtype have higher 3-year pathologic complete response (pCR) rates when treated with trastuzumab and chemotherapy compared to chemotherapy alone. Classification was performed using PAM50, a diagnostic test based on the expression levels of 50 genes \citep{parker2009supervised}. However, the specific genes driving these differences in response to trastuzumab within the HER2+ population remain unclear. 

\begin{table}[b]
\centering
\scalebox{0.8}{
\begin{tabular}{lcccccccccc}
\toprule
\textbf{Method}  & BAG1 & CDC20 & ERBB2 & MYC  & ACTR3B &  BLVRA & CCNE1 &  FGFR4  & SFRP  \\
\midrule
Aggregate  & 0.034  & 0.647 & 0.447& 0.917 & 0.210 & 0.177 & 0.556  & 0.990 & 0.849 \\
\addlinespace
Cauchy  & 0.031   & 0.283 & 0.764 & 0.409 & 0.944 &  0.971 & 0.814 &  0.350 &  0.551  \\
\addlinespace
Imai et al.\ (2024)  & 0.124   & 0.149 & 0.449 & 0.332  & 0.397  & 0.355 & 0.360 & 0.715 &  0.210  \\
\bottomrule
\end{tabular}}
\caption{$p$-values from the analysis based on the NOAH study. The null hypothesis is that there is a constant average treatment effect across values of each of the nine genes.}
\label{tab:2}
\end{table}

We apply our proposed test to investigate whether there is evidence against effect homogeneity with respect to gene expression biomarkers. We use data from the NOAH trial (GEO accession GSE50948), which includes 334 participants, 235 of whom were HER2+. The HER2-positive patients were randomized to receive either trastuzumab with neoadjuvant chemotherapy $(T=1)$ or neoadjuvant chemotherapy alone $(T=0)$. Baseline tumor gene expression levels passing quality control were available for $n=111$ HER2+ patients. The primary outcome is pCR within three years, defined as the absence of invasive breast cancer after treatment.

\begin{figure}[t]
    \centering
    \includegraphics[width=1\textwidth]
    %{Attachments/HERS_plot.pdf}
    {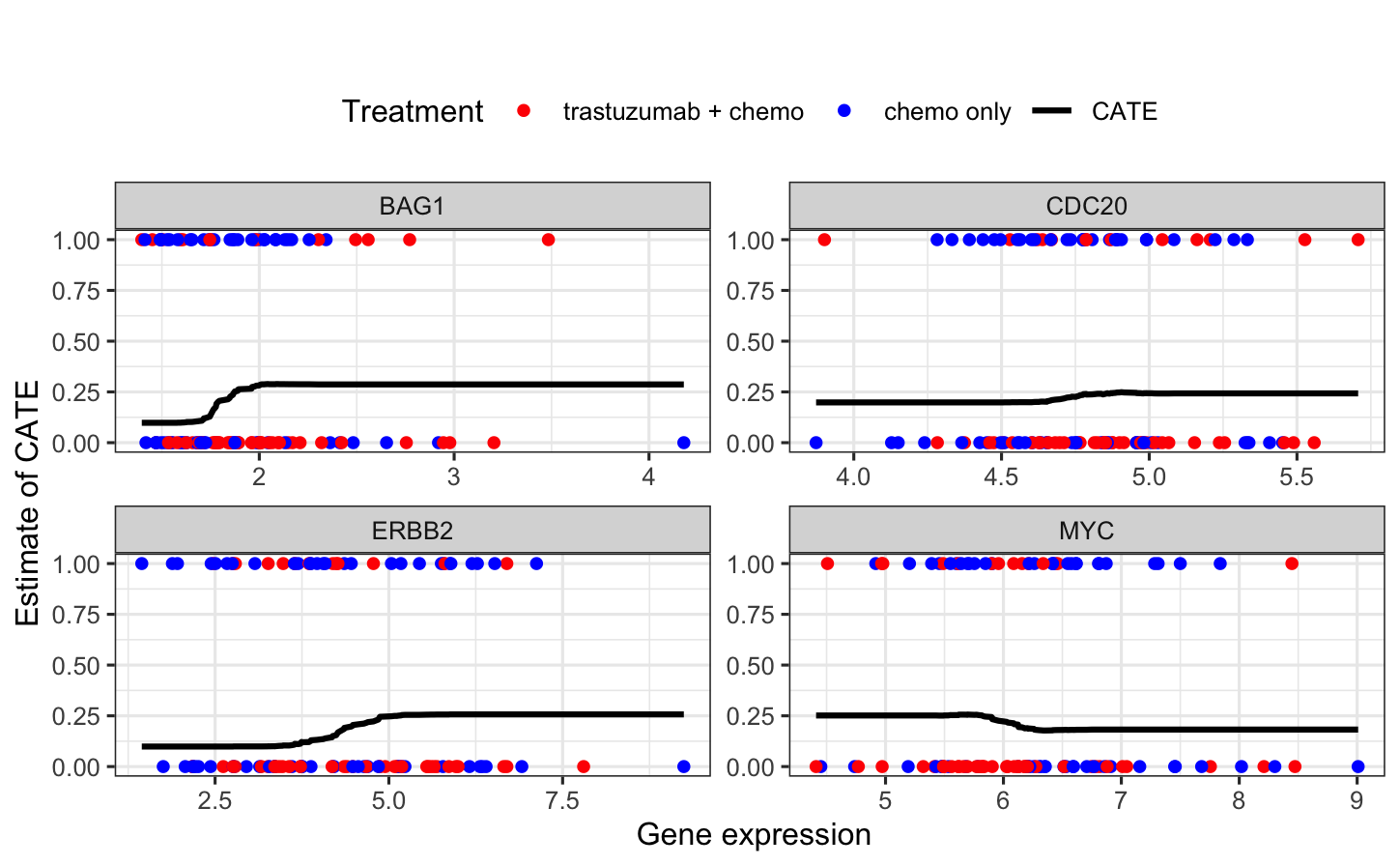}
    \caption{ Estimate of CATE using causal forest \citep{athey2019machine} for each of the four genes (BAG1, CDC20, ERBB2, MYC).}
    \label{fig4}
\end{figure}

Following \cite{roth2018framework}, we consider the nine genes identified by \cite{prat2014based} as demonstrating predictive power comparable to the full 50-gene panel used in the PAM50 test. Table~\ref{tab:2} presents the $p$-values from the proposed tests (Aggregate and Cauchy) alongside those from the test of 
\cite{imai2024statistical}. Results for four genes are presented here, 
with the remaining five genes discussed in Supplementary~\ref{sec:additional-data}. The Aggregate and Cauchy tests are implemented using the same grid of 
$\lambda$ values as in the simulation study. Among the nine genes, the proposed 
tests yield $p$-values below $0.05$ for the BAG1 gene, suggesting evidence 
of heterogeneity. In contrast, the test of \cite{imai2024statistical} 
yields its lowest $p$-value of $0.100$ for the ERBB2 gene. All other genes 
exhibit large $p$-values across all three tests.

Figure~\ref{fig4} illustrates the estimated CATE as a function of expressions of four selected genes. BAG1 shows strong evidence of a non-constant CATE, suggesting the presence of treatment effect heterogeneity. This pattern indicates that patients with relatively low BAG1 expression may benefit less from 
trastuzumab compared with chemotherapy alone, whereas patients with high BAG1 expression may derive greater benefit. Although the CATE curves for CDC20, ERBB2, and MYC appear non-constant, our tests did not find them statistically significant.

\section{Discussion \label{conclusion}}

We presented a general nonparametric test for assessing whether the evaluation of a smooth functional at a conditional distribution is constant in the conditioning variables. We develop a test that is based on representing norms as a supremum taken over a collection of finite-dimensional statistical parameters indexed by functions in a prescribed function class. Our approach ensures control of the type-I error rate for any appropriately selected smooth function class. Moreover, implementing our test only requires the derivation of the first-order influence functions, and we showed that, in many cases, influence functions can be readily obtained by leveraging existing results. Specifically, we established results that provide the explicit form of the EIF when the statistical functional of interest is a linear transformation of the conditional mean, or when the statistical functional does not depend on the distribution of the conditioning variables.

The choice of function class $\mathcal{H}$ can impact the power of the test. To improve the power, we proposed an aggregate test that combines test statistics obtained from different function classes. While the aggregate test offers improved power, the optimal choice of function class, or the tuning parameter for the RKHS class, remains a challenging problem in nonparametric testing settings. 

Our work offers several future research directions. Firstly, our results require that the nuisance functions and their estimators are not overly complex; that is, they must belong to a Donsker class. This may preclude use of some common data-adaptive estimators such as random forests and neural network. One approach to relaxing this condition is constructing a cross-fitted  version of our estimators \citep{chernozhukov2018double}. Secondly, in this work, we implement our test using two function classes: the class of thresholded indicators and RKHS induced by the Sobolev space; however, other function classes such as functions with bounded variation norm, functions with finite VC dimension (e.g., finite-depth trees), and RKHS with other kernels, could also be used. A formal analysis of how power is affected by different choices of function classes would be an interesting topic of future research.

\section*{Acknowledgements}
The authors gratefully acknowledge funding from the U.S. National Institutes of Health. 
\section*{Additional Information}
 The \texttt{R} code for implementing the proposed test (Algorithm~\ref{Algo1} \& \ref{Algo2}) along with the examples presented in the paper is available at \url{https://github.com/albertosom/NPtest}.

\bibliographystyle{apalike}
\bibliography{reference}

\clearpage
%\section{Supplementary Material \label {supp}}

\begin{center}
    \Large Supplemental Materials to ``\emph{A general nonparametric framework for testing hypotheses about function-valued parameters''}
\end{center}
\renewcommand{\thesection}{S\arabic{section}}
\setcounter{page}{1}
\setcounter{section}{0}
\renewcommand{\thefigure}{S.\arabic{figure}}

\setcounter{figure}{0}

\renewcommand{\thetable}{S\arabic{table}}

\renewcommand{\thelemma}{S\arabic{lemma}}

\renewcommand{\theproposition}{S\arabic{proposition}}

\setcounter{figure}{0}

\section{Details on the reproducing kernel Hilbert space function (RKHS) class } \label{sec:rkhs}
Let $\mu$ be a Lebesgue measure over $\mathcal{X}$, a compact subset of $\mathbb{R}^{d}$ and consider the function class $L^{2}(\mathcal{X};\mu)= \left \{f(x) \in\mathbb{R}: \int f^2(x)d\mu < \infty \right \}$.  We denote  $\mathbb{K}$ as a positive semidefinite universal kernel function and $\mathbb{H}$ the associated reproducing kernel Hilbert space. A kernel $\mathbb{K}$ is called a \textit{Mercer kernel} if it is continuous and $\int_{\mathcal{X}\times \mathcal{X}} \mathbb{K}^{2}(x_1,x_2)d\mu(x_1)d\mu(x_2) <\infty$. For a given a Mercer kernel, $\mathbb{K}$, there exist a set of eigenfunctions  $(\psi_j)^{\infty}_{j=1}$,  and non-negative eigenvalues $(\lambda_{j})^{\infty}_{j=1}$ that induces $\mathbb{H}$ given below as
\begin{align*}
        \mathbb{H}= \left \{ h=\sum_{j=1}^{\infty}\alpha_{j}\psi_{j} \quad  | \quad \sum_{j=1}^{\infty}\alpha_{j}^{2} < \infty \text{ and }\sum_{j=1}^{\infty}\frac{\alpha_{j}^2}{\lambda_j} < \infty  \right\}.
    \end{align*}

The kernel has the representation $\mathbb{K}(x_1, x_2) = \sum_{j=1}^\infty \lambda_j \phi_j(x_1)\phi_j(x_2)$  and the space is endowed with inner product $\langle f,g \rangle_{\mathbb{H}}=\sum_{i=1}^{\infty}\frac{\langle f,\psi_{i}\rangle\langle f,\psi_{i}\rangle}{\lambda_i}$, where $\langle \cdot,\cdot\rangle$ denote the inner product in $L^{2}(\mathcal{X};\mu)$. The smoothness of any function $h \in \mathbb{H}$ can be measured by the RKHS norm
\begin{align*}
   \|h\|^2_{\mathcal{H}} =\langle h,h \rangle_{\mathbb{H}} =\sum_{i=1}^\infty \frac{a^2_i}{\lambda_i}.
\end{align*}
As presented in Section~\ref{RKHS_class}, we consider the RKHS induced by the eigenfunctions $\psi_{2j-1}(x) = \sqrt{2}\cos(2j\pi x)$ and $\psi_{2j}(x) = \sqrt{2}\sin(2j\pi x)$ with their corresponding eigenvalues $\lambda_{2j-1}=\lambda_{2j}=(2j\pi)^{-4}$ for $j=1,2,\ldots$. For some large positive integer $D$, we define the truncated version of $\mathbb{H}$ as
\begin{align*}
   \mathbb{H}_{D} = \Bigg \{\sum_{j=1}^D a_j\psi_j(x) : a_j \in \mathbb{R}, \;  \bm{a}^T\Gamma \bm{a} < \beta, \text{ for some fixed } \beta>0 \Bigg \}
\end{align*}
where $\Gamma$ is the same as defined in Section~\ref{RKHS_class}. The constraint $ \bm{a}^T\Gamma \bm{a} < \beta$ only controls the roughness of the function but not their shape complexity. Large values of $\beta$ relaxes the smoothness penalty, allowing functions with larger magnitude, while small values of $\beta$ shrinks functions in the function class toward zero; however, in both cases, the shape of the functions remains essentially unchanged. To construct a test with power against a broad class of alternatives, we require the function class to allow variation in shape. In particular, we want constraints on the function class that affect the sampling distribution of the test statistic through its distributional shape rather than through its scale. It is also beneficial to restrict shapes that are likely to only improve power against implausible alternatives. Therefore, we introduce a variance constraint on the estimated functions in the RKHS. 

Let $\mathbf{V}$ is the $D\times D$ matrix with entry $(d_1,d_2)$ given as
\begin{align*}
      \mathbf{V}_{d_1,d_2}= n^{-1}\sum_{i=1}^n \left\{ \psi_{d_1}(x_i)-n^{-1}\sum_{i=1}^n\psi_{d_1}(x_i) \right \} \left \{   \psi_{d_2}(x_i)-n^{-1}\sum_{i=1}^n\psi_{d_2}(x_i) \right\}.
  \end{align*}
  The empirical variance for each function $\sum_{j=1}^D a_j\psi_j(x)$ is $\bm{a}^T\mathbf{V} \bm{a} $. Putting all together, our function class that balances smoothness and variance constraint is 
\begin{align*}
\mathcal{H} = \Big \{\sum_{j=1}^D a_j\psi_j(x) : a_j \in \mathbb{R}, \; \gamma \bm{a}^T\Gamma \bm{a} + (1-\gamma)\bm{a}^T\mathbf{V} \bm{a} \le C, \text{ for some fixed } C>0, \text{ and } \gamma\in [0,1] \Big \}
\end{align*}

\newpage
\section{Additional results from data application}\label{sec:additional-data}
In this section, we present plot of estimate of the CATE for the other five genes (ACTR3B, BLVRA, CCNE1, FGFRA, SFRP1). We observe that the CATE curves are all almost flat, suggesting no differential effect of each gene on treatment response.

\begin{figure}[h]
    \centering
    \includegraphics[width=1\textwidth]
    %{Attachments/HERS_plot.pdf}
    {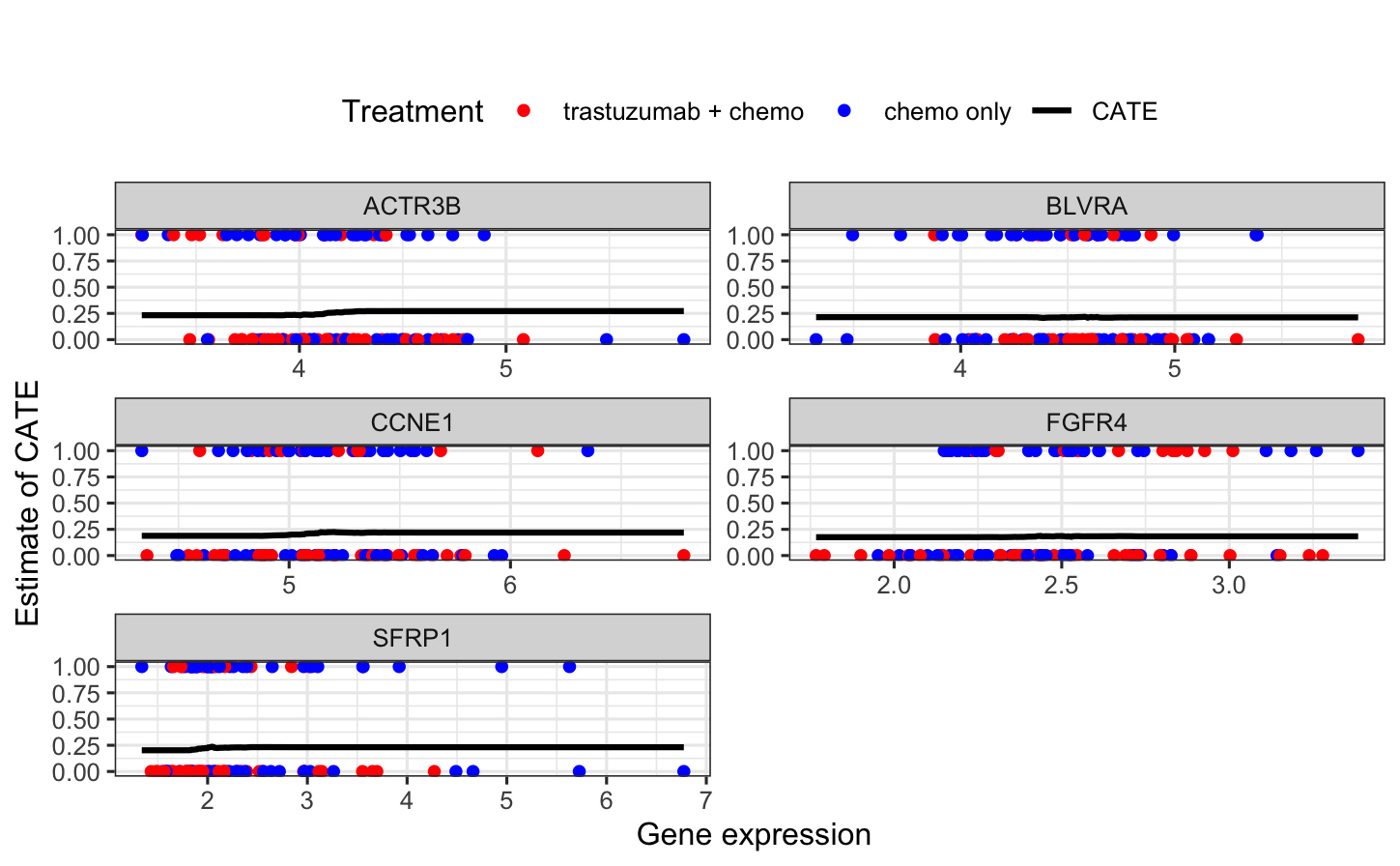}
    \caption{ Estimate of CATE using causal forest \citep{athey2019machine} for each of the five genes (ACTR3B, BLVRA, CCNE1, FGFRA, SFRP1).}
    \label{fig4}
\end{figure}

\section{ Discussion of additional examples}\label{sec:additional-examples}
% \subsection{Additional Examples}
In this section, we discuss two additional examples: testing the constancy of the conditional covariance and assessing treatment effect heterogeneity in survival settings. For each example, we show that the testing problem fall within our hypotheses class, derive the corresponding efficient influence function (EIF) and one-step estimator needed to construct the test. We also present numerical studies for the hypothesis test concerning the constancy of the conditional covariance.

\subsection{Example 3 \normalfont{(Testing constancy of conditional covariance)}}
\setcustomref{examp3}{3}  % Manually set reference to 3

\noindent Inference on the conditional covariance function appears naturally in many scientific studies. For example, in psychometric studies, it is of interest to assess whether there is significant evidence to support that the symmetric relationship between speed and accuracy (assumed to be linear) depend on other continuous variables such as age and attention when performing cognitive task \citep{tu2024cocoa}. Such hypothesis can be framed as testing whether the conditional covariance function is constant. To state this formally, let $O=(X,Y,Z) \sim P_{0}$, where $X, Y \in \mathbb{R}$ are two we wish to study their association, and $Z \in \mathbb{R}^{d}$ is a set of covariate. When $d=1$, we particularly interested in the setting where $Z$ is continuous. We define our smooth map as $\Psi(P)=\int (x-\int xdP)(y-\int y dP) dP$ for $P \in \mathcal{M}$, then the function-valued and null parameters are   $\Psi_{0,z}=\E[(Y-\E[Y|Z=z])(X-\E[X|Z=z])|Z=z]$ and $\theta_0 = \E[\E[(Y-\E[Y|Z])(X-\E[X|Z])|Z]]$, respectively. The null hypothesis of constancy of the conditional covariance function can be expressed as 
\begin{align*}
    H_0: \E[(Y-\E[Y|Z=z])(X-\E[X|Z=z])|Z=z]=\theta_{0} \quad P_{0,Z}\text{-almost surely.}
\end{align*}
\noindent For the special case where $\theta_{0}=0$, the conditional independence test based on the generalized covariance measure \citep{shah2020hardness} can be used. However, for our scientific question of interest, $\theta_{0}$ is often not zero and hence the test by \cite{shah2020hardness} generally does not apply.

\noindent To implement our test, we to derive the EIF, $D^{*}_{h,0}(o)$, of $\Omega_{P_{0}}(h)$. Lets denote the finite-dimensional parameter from the marginalization of $\Psi_{0,z}$ as $\widetilde{\Psi}_{0}=\E[\E[(Y-\E[Y|Z])(X-\E[X|Z])|Z]]$, then $\widetilde{\Psi}_{0}$ is pathwise differentiable and its EIF is given in \cite{xiang2020flexible} and \cite{kennedy2024semiparametric}  as
\begin{align*}
    \widetilde{D}_{P_{0}}(o) & = (y-\E[Y|Z=z])(x-\E[X|Z=z]) - \widetilde{\Psi}_{0}.
\end{align*}

Following Proposition~\ref{prop2}, the EIF,  $D^{*}_{h,0}$, can be expressed as,
\begin{align*}
       D^{*}_{h,0}(o)&=\Big ((y-\mu_{0,y|z})(x-\mu_{0,x|z})-\widetilde{\Psi}_{0} \Big)h^{c}_{P_{0,Z}}(z)
       - \int \Psi_{0,Z}h^{c}_{P_{0,Z}}(z)dP_{0,Z}(z),
\end{align*}
and the one-step estimator, $\widehat{\Omega}^{os}_{n}(h)$, of $\Omega_{P_{0}}(h)$ is
\begin{align*}
        \widehat{\Omega}^{os}_{n}(h)&= n^{-1} \sum_{i=1}^n \left[ \left(Y_{i}-\mu_{n,y|z}(Z_{i})\right)\left(X_{i}-\mu_{n,x|z}(Z_{i})\right)- \mathbb{P}_{n} \Psi_{n,z}\right ] h^{c}_{n}(Z_{i}), 
   \end{align*}
   where  $\bar{h}_{n}=\frac{1}{n}\sum_{i}^{n}h(Z_{i})$, $h^{c}_{n}(Z)= h(Z)-\bar{h}_{n}$, $ \mathbb{P}_{n} \Psi_{n,z}= n^{-1}\sum_{j=1}^{n}\left(Y_{j}-\mu_{n,y|z}(Z_{i})\right)\left(X_{j}-\hat{\mu}_{0,x|z}(Z_{i})\right)$, $\mu_{0,y|z}=\E[Y|Z=z]$, $\mu_{0,x|z}=\E[X|Z=z]$ and their corresponding estimators as $\mu_{n,y|z}$ and $\mu_{n,x|z}$, respectively. 

\subsection{Simulation Results for Example 3}

\begin{figure}[t]
    \centering
    \includegraphics[width=0.90\textwidth]{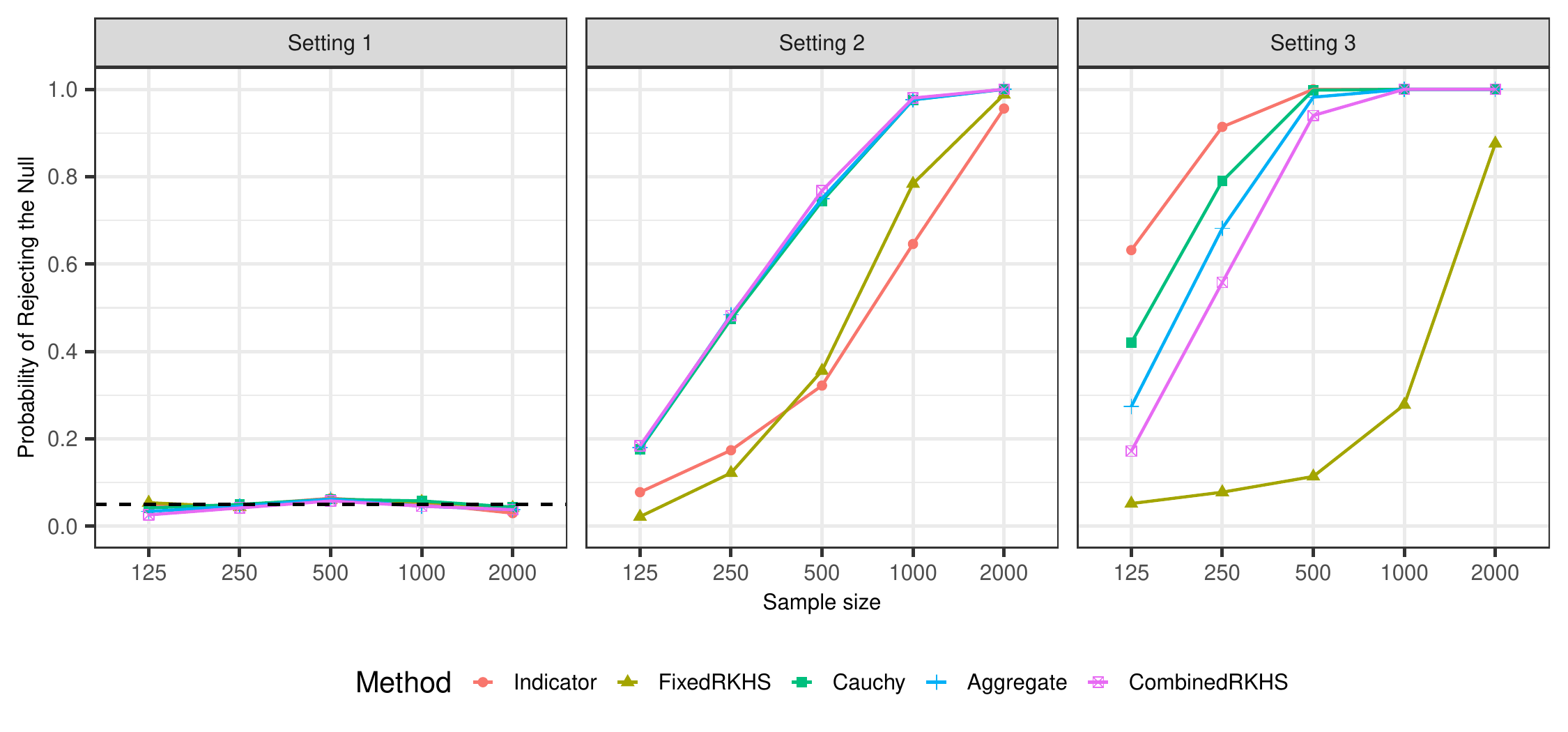}
    \caption{Empirical probability of rejection for hypotheses in Example 3 under different data generating mechanisms and sample sizes. The horizontal dashed line indicates the nominal 0.05 significance level.}
\label{fig3}
\end{figure}

We study the finite-sample performance of our proposed test for constancy of the conditional covariance function. In particular, we consider three settings to evaluate both type-I error control and power. For all the settings we implement our test using the same function classes presented in Section \ref{aggregatetest}. We also estimate the nuisance functions $\mu_{0,y|z}$ and $\mu_{0,x|z}$ using generalized additive models \citep{hastie1990generalized} with natural cubic spline components, implemented in the R package \texttt{mgcv}. 

\begin{itemize}
    \item \textbf{Setting 1:} $(Y,X) \sim N \left( \begin{psmallmatrix}
0 \\
0
\end{psmallmatrix}, \begin{psmallmatrix}
1 & 0 \\
0 & 1 
\end{psmallmatrix} \right)$ and $Z\sim \text{Unif}(-1,1)$.
    
    \item \textbf{Setting 2:} $Z \sim \text{Unif}(0,1)$ and $(Y,X \mid Z=z) \sim N\left( \begin{psmallmatrix}
0 \\
0
\end{psmallmatrix}, \begin{psmallmatrix}
1 & \rho(z) \\
\rho(z) & 1 
\end{psmallmatrix} \right)$, where $\rho(z)=\frac{e^{z^2}-1}{e^{z^2}+1}$.
    
    \item \textbf{Setting 3:} $Z \sim \text{Unif}(0,1)$, $X \sim N(0,1)$, and $Y = 0.5X\mathbbm{1}(Z > 0) + \varepsilon$, where $\varepsilon \sim N(0,1)$.
\end{itemize}

\noindent The results are presented in Figure~\ref{fig3}. Under Setting~1, where the null hypothesis of constant conditional covariance holds, all versions of the proposed test control the type-I error well as the sample size increases, with only minor fluctuations around the nominal level of 0.05. Under Settings~2 and 3, where the conditional covariance varies with $Z$, we observe increasing power as the sample size grows across all choices of function classes. The relative performance across function classes reflects how well each class captures the underlying structure of the conditional covariance. In Setting~2, the conditional covariance function is nonlinear in $Z$. As a result, tests constructed using more flexible function classes achieve higher power compared to simpler classes, such as the thresholded indicator class, which cannot adequately approximate the nonlinear structure. In contrast, Setting~3 exhibits a simple conditional relationship driven by an indicator function in $Z$ which is monotone. In this case, the test based on the thresholded indicator function class performs better in smaller samples, as this class aligns closely with the true underlying structure. Despite these differences, the aggregated test demonstrates consistently strong performance across all settings, highlighting the advantage of combining multiple function classes to adapt to a range of alternative structures.

\subsection{Example 4 \normalfont{(Assessing treatment effect heterogeneity in survival setting)}}
\setcustomref{examp4}{4}  % Manually set reference to 3
n this example, we describe how our proposed framework can be applied to assess effect heterogeneity in a setting with right-censored time-to-event outcomes. Recall in Example~\ref{examp2}, we applied our test to assess treatment effect heterogeneity in settings with binary or continuous outcome. Handling time-to-event data can be challenging due to censoring; therefore, we discuss this setting here as an additional example. Consider the ideal unobservable data vector $O_{F}=(W,A,T^{(0)},T^{(1)},C^{(1)},C^{(0)})$. Here, $W\in \mathcal{W} \subset \mathbb{R}^{p}$ is a vector of baseline covariate, $A \in \{0,1\}$ is the binary treatment variable, and for $a\in \{0,1\}$, $T^{(a)} \in (0,\infty]$ and $C^{(a)} \in (0,\infty]$ denote the potential event time and censoring time under treatment assignment $A=a$. We do not observe the full data vector $O_F$, but instead we observe $O=(W,A,Y,\Delta) \sim P_{0}$, where $T$ and $C$ are the observed event and censoring times, $Y=\text{min}\{ T,C\}$, and $\Delta= \mathbbm{1}\{T\le C \}$. For $s \subset \{1,\ldots,p\}$, let $W_s$ be the sub-vector containing elements of $W$ with indices in $s$.  Consider $t \in (0,\tau]$ for some positive $\tau < \infty$, we are interested in inference about the counterfactual conditional survival curve
\begin{align*}
     \Psi_{0,w_{s}}&=P_{0}(T^{(1)}>t|W=w)-P_{0}(T^{(0)}>t|W_{s}=w_{s}).
\end{align*}
To assess treatment effect heterogeneity under the time-to-event outcome, we test the null hypothesis:
\begin{align} \label{Hetcondsurv}
   H_{0}:  \Psi_{0,w_{s}} =\theta_{0} \quad P_{0,W_{s}}\text{-almost surely}
\end{align}
where $\theta_{0}= P_{0}(T^{(1)}>t)-P_{0}(T^{(0)}>t)$. We also see here that the above null hypothesis fall within our general class of hypotheses. Under the assumptions stated \citep{westling2024inference}, we can identify $\Psi_{0,w_{s}}$  in terms of the distribution $P_{0}$ of the observed data unit. Let $S_{0}(t|A=a,W=w)=P_{0}(T>t|A=a,W=w)$ denote the conditional survival function corresponding to treatment $A=a$, then the $\Psi_{0,w_{s}}$ is identified as
\begin{align*}
    \Psi_{0,w_{s}}&=\E[S_{0}(t|A=1,W)-S_{0}(t|A=0,W)|W_{s}=w_{s} ].
\end{align*}
To construct a test for \eqref{Hetcondsurv} using our proposed framework, we first have to derive $D^{*}_{h,0}(o)$, the the EIF of $\Omega_{P_{0}}(h)=\int  \Psi_{0,w_{s}} (w_{s})h_{P_{0,W_{s}}}^{c}(w_{s})dP_{0,W_{s}}(w_{s})$. Recall that the general form of $D^{*}_{h,0}$ depends on $D^{w}_{P_{0}}$, the conditional influence function of $\Psi_{0,w_{s}}$. In this example we again show that deriving $D^{*}_{h,0}$ can be streamlined. Let $\eta_{0,w}=S_{0}(t|A=1,W=w)-S_{0}(t|A=0,W=w)$, then $\Psi_{0,w_{s}}=\E[\eta_{0,w}|W_{s}=w_{s}]$. Following similar argument in \cite{westling2024inference}, the conditional influence function of $\eta_{0,w}$ is 
\begin{align*}
  \sum_{a_{0}=0}^{1}  -(-1)^{a_{0}} S_{0}(t|a_{0},w) \Big [1- \frac{\mathbbm{1}\{a=a_{0}\}}{\pi_{0}(a_{0}|w)}  \Big\{\frac{\mathbbm{1}\{y \le t,\delta=1\}}{S_{0}(u|a_{0},w)G_{0}(u|a_{0},w)}-\int_{0}^{t\wedge y}\frac{\Lambda_{0}(du|a,w)}{S_{0}(u|a_{0},w)G_{0}(u|a_{0},w)} \Big\} \Big ], 
\end{align*}

\noindent where $G_{0}(u|a,w)\coloneqq   \prodi_{(0,u]} \{1-H_{0}(ds|a,w)\}$, $H(u|a,w)=\int_{[0,u] } \Big \{\frac{S_{0}(u-|a,w)}{S_{0}(s|a,w)} \frac{F_{0.0}(ds|a,w)}{R_{0}(s|a,w)}\Big \}$,  $F_{0.0}(u|a,w)=P_{0}(Y\le u,\Delta=0 |A=a,W=w)$, $F_{0,1}(t|a,w)=P_{0}(Y \le t,\Delta=1|A=a,W=w)$, $R_{0}(t|a,w)=P_{0}(Y\ge t|A=a,W=w)$, $S_{0}(t|a,W)\coloneqq \prodi_{(0,t]}\{1-\Lambda_{0}(du|a,w) \}$, $\Lambda_{0}(t|a,w) =\int_{0}^{t}F_{0,1}(du|a,w)/R_{0}(u|a,w)$, and $\prodi$ denotes the Riemann-Stieltjes product integral. We then obtain $D^{*}_{h,0}$ as
\begin{align*}
     &\Bigg \{ \sum_{a_{0}=0}^{1}  -(-1)^{a_{0}} S_{0}(t|a_{0},w) \Big [1- \frac{\mathbbm{1}\{a=a_{0}\}}{\pi_{0}(a_{0}|w)}  \Big\{\frac{\mathbbm{1}\{y \le t,\delta=1\}}{S_{0}(u|a_{0},w)G_{0}(u|a_{0},w)} 
      -\int_{0}^{t\wedge y}\frac{\Lambda_{0}(du|a,w)}{S_{0}(u|a_{0},w)G_{0}(u|a_{0},w)} \Big\} \Big ] \\ 
      & - \int  \Psi_{0,w_{s}} (w_{s})dP_{0,W_{s}}(w_{s}) \Bigg \} h_{P_{0,W_{s}}}^{c}(w_{s}) - \Omega_{P_{0}}(h).
\end{align*}

\section{ Proofs of theoretical results} \label{sec:proofs}

\subsection{Proof of Theorem~\ref{theorem1}}
\begin{proof}
We show that $\Omega_{P_{0}}(h)$ is pathwise differentiable and present the form of its EIF, $D^{*}_{h,0}$. That is, for a path $dP_t=(1+ts)dP_0$ through $P_{0} \in \mathcal{M}$ in the direction $s$ where $\int s(o)dP_{0}(o)=0$ and $\int[s^2(o)]dP_{0}(o)<\infty$, there exist some $ D^*_{h,0}(o)$ such that for a fixed $h$,
 \begin{align*}
       \frac{\partial}{\partial t} \Omega_{P_{t}}(h) \Bigg|_{t=0} = \int  D^*_{h,0}(o) s(o) dP_{0}(o),
 \end{align*}
 where  $D^*_{h,0}(o)$ is uniformly bounded satisfying $ \int D^*_{h,0}(o)dP_{0}(o)=0$ and $\int (D^*_{h,0}(o))^{2}dP_{0}(o)<\infty$.

 Recall that $\Omega_{P_{t}}(h)= \int  \Psi_{t,v}(v)h^{c}_{P_{t,V}}(v) dP_{t,V}(v)$ where $h^{c}_{P_{t,v}}(v)=h(v) -\int h(v)dP_{t,V}(v)$. Now suppose conditions \hyperref[C1]{\rm{C1}} and \hyperref[C2]{\rm{C2}} are satisfied, then 

\begin{align*}
    \frac{\partial}{\partial t} \Omega_{P_{t}}(h) \Bigg|_{t=0} &=\ \frac{\partial}{\partial t} \int  \Psi_{t,v}(v)h^{c}_{P_{t,V}}(v) dP_{t,V}(v)\Big|_{t=0} \\
    &= \int  \frac{\partial}{\partial t} \Psi_{t,v}(v)\Big|_{t=0} h^{c}_{P_{0,V}}(v) dP_{0,V}(v) + \int  \Psi_{0,v}(v)  \frac{\partial}{\partial t}h^{c}_{P_{t,V}}(v) \Big|_{t=0} dP_{0,V}(v)\\
    &\quad+ \int  \Psi_{0,v}(v)h^{c}_{P_{0,V}}(v) \frac{\partial}{\partial t}dP_{t,V}(v) \Big|_{t=0}.
\end{align*}

Under the full-data path $dP_t=(1+ts)dP_0$, the induced score for the conditional distribution $P_{t,O|V=v}$ is 
\[
s_v(o)=s(o)-E_{0}[s(O)\mid V=v].
\]
By Condition \hyperref[C2]{\rm{C2}},
\begin{align*}
\frac{\partial}{\partial t}\Psi_{t,v}(v)\Bigg|_{t=0}
&= \int D^{v}_{P_0}(o)s_v(o)dP_{0,O|v}(o) \\
&= \int D^{v}_{P_0}(o)\{s(o)-E_0[s(O)\mid V=v]\}dP_{0,O|v}(o).
\end{align*}
Since $\int D^{v}_{P_0}(o)dP_{0,O|v}(o)=0$, this simplifies to
\[
\frac{\partial}{\partial t}\Psi_{t,v}(v)\Bigg|_{t=0}
= \int D^{v}_{P_0}(o)s(o)dP_{0,O|v}(o).
\]
Therefore,
\begin{align*}
\int  \frac{\partial}{\partial t} \Psi_{t,v}(v)\Bigg|_{t=0} h^{c}_{P_{0,V}}(v) dP_{0,V}(v)
= \int D^{v}_{P_0}(o) h^{c}_{P_{0,V}}(v) s(o) dP_{0}(o).
\end{align*}

Next, since $h^{c}_{P_{t,V}}(v)=h(v)-\int h(v)dP_{t,V}(v)$,
\begin{align*}
\frac{\partial}{\partial t}h^{c}_{P_{t,V}}(v)\Big|_{t=0}
&= -\frac{\partial}{\partial t}\int h(v)dP_{t,V}(v)\Big|_{t=0} \\
&= -\int h(v)s(o)dP_{0}(o).
\end{align*}
Since $\int s(o)dP_0(o)=0$, we may replace $h(v)$ by $h^{c}_{P_{0,V}}(v)$, giving
\[
\frac{\partial}{\partial t}h^{c}_{P_{t,V}}(v)\Big|_{t=0}
= -\int h^{c}_{P_{0,V}}(v)s(o)dP_{0}(o).
\]
Hence,
\begin{align*}
\int  \Psi_{0,v}(v)  \frac{\partial}{\partial t}h^{c}_{P_{t,V}}(v) \Bigg|_{t=0} dP_{0,V}(v)
= - \int \Psi_{0,v}(v)dP_{0,V}(v) \int h^{c}_{P_{0,V}}(v)s(o)dP_{0}(o).
\end{align*}

Finally, for the third term,

\begin{align*}
\int  \Psi_{0,v}(v)h^{c}_{P_{0,V}}(v) \frac{\partial}{\partial t}dP_{t,V}(v) \Bigg|_{t=0}
= \int \Psi_{0,v}(v)h^{c}_{P_{0,V}}(v)s(o)dP_{0}(o).
\end{align*}

Combining all terms,
\begin{align*}
\frac{\partial}{\partial t} \Omega_{P_{t}}(h) \Bigg|_{t=0}
&= \int \Bigg(\left[D^{v}_{P_{0}}(o) -  \int \Psi_{0,v}(v)dP_{0,V}(v)  + \Psi_{0,v}(v)\right]h^{c}_{P_{0,V}}(v)\Bigg) s(o)dP_{0}(o).
\end{align*}

Thus, the efficient influence function is
\begin{align*}
D^*_{h,0}(o)= \Bigg (D^{v}_{P_{0}}(o) -\int\Psi_{0,v}(v)dP_{0,V}(v) +   \Psi_{0,v}(v) \Bigg)
h^{c}_{P_{0,V}}(v)-  \Omega_{P_{0}}(h).
\end{align*}
This completes the proof.
\end{proof}

\subsection{Proof of Proposition~\ref{prop}}
\begin{proof}
We want to show that $D^*_{h,0}$ has an explicit form when $\Psi_{0,v}$ is a continuous linear functional of the conditional mean function. We first obtain the  conditional influence function \(D^v_{P_0}\) for \(\Psi_{0,v}\), and then substitute it into the general form of the  EIF, $D^*_{h,0}$.

Fix \(v\), and consider a path through the conditional law \(P_{0,O\mid v}\) in the direction \(s_v\), given by
\[
dP_{t,O\mid v}=(1+t s_v)dP_{0,O\mid v},
\]
where \(P_{0,O\mid v}=P_{t=0,O\mid v}\), and \(s_v\) satisfies
\[
\int s_v(o)\,dP_{0,O\mid v}(o)=0,
\qquad
\int s_v^2(o)\,dP_{0,O\mid v}(o)<\infty.
\]

Differentiating with respect to \(t\) at \(t=0\), we obtain
\begin{align*}
\frac{\partial}{\partial t}\Psi_{t,v}(v)\Bigg|_{t=0}
&=
\frac{\partial}{\partial t}\int m(o,\mu_t(u_{2},v))\,dP_{t,O\mid v}(o)\Big|_{t=0}\\
&=
\int m(o,\mu_0(u_{2},v))
\frac{\partial}{\partial t}dP_{t,O\mid v}(o)\Big|_{t=0}
+
\frac{\partial}{\partial t}\E[m(O,\mu_t(u_{2},v))\mid V=v]\Big|_{t=0}.
\end{align*}

Since
\[
dP_{t,O\mid v}=(1+t s_v)dP_{0,O\mid v},
\]
it follows that
\[
\frac{\partial}{\partial t}dP_{t,O\mid v}(o)\Big|_{t=0}
=
s_v(o)\,dP_{0,O\mid v}(o).
\]
Hence,
\begin{align*}
\int m(o,\mu_0(u_{2},v))
\frac{\partial}{\partial t}dP_{t,O\mid v}(o)\Bigg|_{t=0}
&=
\int m(o,\mu_0(u_{2},v))s_v(o)\,dP_{0,O\mid v}(o)\\
&=
\int \{m(o,\mu_0(u_{2},v))-\Psi_{0,v}(v)\}s_v(o)\,dP_{0,O\mid v}(o),
\end{align*}
where the second equality uses \(\int s_v(o)\,dP_{0,O\mid v}(o)=0\). For the second term,

\begin{align*}
\frac{\partial}{\partial t}\E[m(O,\mu_t(U_{2},V))\mid V=v]\Bigg|_{t=0}
&=
\E\!\left[
\alpha_0(U_{2},V)\E[(U_{1}-\mu_0(U_{2},V))s_v(O)\mid U_{2},V]
\mid V
\right]\\
&=
\E[\alpha_0(U_{2},V)(U_{1}-\mu_0(U_{2},V))s_v(O)\mid V]\\
&=
\int \alpha_0(u_{2},v)\{u_{1}-\mu_0(u_{2},v)\}s_v(o)\,dP_{0,O\mid v}(o).
\end{align*}

Combining the two terms, we obtain
\begin{align*}
\frac{\partial}{\partial t}\Psi_{t,v}(v)\Bigg|_{t=0}
=
\int
\Big\{
m(o,\mu_0(u_{2},v))
-\Psi_{0,v}(v)
+\alpha_0(u_{2},v)\big(u_{1}-\mu_0(u_{2},v)\big)
\Big\}
s_v(o)\,dP_{0,O\mid v}(o).
\end{align*}

Thus, the conditional influence function is
\[
D^v_{P_0}(o)
=
m(o,\mu_0(u_{2},v))
-\Psi_{0,v}(v)
+\alpha_0(u_{2},v)\big(u_{1}-\mu_0(u_{2},v)\big).
\]

Moreover,
\begin{align*}
\E[D^v_{P_0}(O)\mid V]
&=
\E[m(O,\mu_0(U_{2},V))\mid V]-\Psi_{0,v}(v)+
\E[\alpha_0(U_{2},V)\{U_{1}-\mu_0(U_{2},V)\}]\\
&=0,
\end{align*}
since \(\Psi_{0,v}(v)=\E[m(O,\mu_0(U_{2},V))]\) and
\[
\E[U_{1}-\mu_0(U_{2},V)\mid U_{2}, V]=0.
\]
Hence, \(D^v_{P_0}\) satisfies Condition~\hyperref[C2]{\rm C2}.

Replacing \(D^v_{P_0}\) in the general EIF formula, we obtain
\begin{align*}
D^*_{h,0}(o)
&=
\left(
m(o,\mu_0(u_{2},v))
+\alpha_0(u_{2},v)\big(u-\mu_0(u_{2},v)\big)
-\int \Psi_{0,v}(v)dP_{0,V}(v)
\right)
h^c_{P_{0,V}}(v)
-\Omega_{P_0}(h).
\end{align*}
This completes the proof.
\end{proof}

\subsection{Extension of Proposition ~\ref{prop}}
In this section, we provide a brief discussion on how the conclusion of Proposition~\ref{prop} can be extended beyond continuous linear functional of the conditional mean function.  Suppose the function-valued parameter is of the form
\[
\Psi_{0,v} = \E[m(O,\eta_{0,v}) \mid V=v],
\]
where $m(o,\eta_{0,v})$ is a pathwise differentiable function of some conditional quantity $\eta_{0,v}$. If the conditional influence function of $\eta_{0,v}$ satisfying condition \hyperref[C2]{\rm{C2}} exists, we can obtain the form of $D^v_{P_0}$, and hence the close form of the EIF, $D^*_{h,0}$.  We denote the  conditional influence function of $\eta_{0,v}$ by $D^{\eta_{0,v}}_{P_0}$. The form of the conditional influence function of $\Psi_{0,v}$ is given by
\[
D^v_{P_0}(o)
=
m(o,\eta_{0,v})
+
m'(o,\eta_{0,v}) D^{\eta_{0,v}}_{P_0}
-
\Psi_{0,v}(v),
\]
where $m'(o,\eta_{0,v})$ denotes the pathwise derivative of $m(O,\eta_{0,v})$ with respect to $\eta_{0,v}$. Substituting this expression into the EIF of $\Omega_{P_0}(h)$ yields
\[
D^*_{h,0}(o)
=
\left\{
m(o,\eta_{0,v})
+
m'(o,\eta_{0,v}) D^{\eta_{0,v}}_{P_0}(o)
-
\int \Psi_{0,v}(v)\, dP_{0,V}(v)
\right\}
h^c_{P_{0,V}}(v)
-
\Omega_{P_0}(h).
\]
This results applies to a broader class that that of the continuous linear functional of the conditional mean function. For example, we can straightforwardly apply our framework to function-valued parameters where $\eta_{0,v}$ is the conditional survival function since the form of the  conditional influence function of the $\eta_{0,v}$ is known in the literature.

\subsection{Proof of Proposition~\ref{prop2}}

\begin{proof}

Recall that
\begin{align*}
\Psi_{0,v}(v)=
\int g(o;P_{0,O\mid v}) dP_{0,O|v}(o) , \text{ and } \widetilde{\Psi}_{0} =
\int \Psi_{0,v}(v)\,dP_{t,V}(v).
\end{align*}

Take a path through \(P_0\) in the direction \(s\), given by
\begin{align*}
dP_t=(1+t s)\,dP_0.
\end{align*}

Differentiating at \(t=0\), we obtain
\begin{align*}
\frac{\partial}{\partial t}\widetilde{\Psi}_t\Bigg|_{t=0}
&=
\frac{\partial}{\partial t}\int \Psi_{t,v}(v)\,dP_{t,V}(v)\Bigg|_{t=0}\\
&=
\int \frac{\partial}{\partial t}\Psi_{t,v}(v)\Bigg|_{t=0}\,dP_{0,V}(v)
+
\int \Psi_{0,v}(v)\frac{\partial}{\partial t}dP_{t,V}(v)\Bigg|_{t=0}.
\end{align*}

By Condition~\hyperref[C2]{\rm C2}, for each fixed \(v\),
\begin{align*}
\frac{\partial}{\partial t}\Psi_{t,v}(v)\Bigg|_{t=0}
&=
\int D^{v}_{P_0}(o)s_v(o)\,dP_{0,O\mid v}(o),
\end{align*}
where
\begin{align*}
s_v(o)=s(o)-E_0[s(O)\mid V=v].
\end{align*}

It follows that
\begin{align*}
\frac{\partial}{\partial t}\Psi_{t,v}(v)\Bigg|_{t=0}
&=
\int D^{v}_{P_0}(o)\{s(o)-E_0[s(O)\mid V=v]\}\,dP_{0,O\mid v}(o)\\
&=
\int D^{v}_{P_0}(o)s(o)\,dP_{0,O\mid v}(o).
\end{align*}
Therefore,
\begin{align*}
\int \frac{\partial}{\partial t}\Psi_{t,v}(v)\Bigg|_{t=0}\,dP_{0,V}(v)
&=
\int \left\{\int D^{v}_{P_0}(o)s(o)\,dP_{0,O\mid v}(o)\right\}dP_{0,V}(v)\\
&=
\int D^{v}_{P_0}(o)s(o)\,dP_0(o).
\end{align*}

\begin{align*}
\int \Psi_{0,v}(v)\frac{\partial}{\partial t}dP_{t,v}(v)\Bigg|_{t=0}
&=
\int \Psi_{0,v}(v)E_0[s(o)\mid V=v]\,dP_{0,V}(v)\\
&=
\int \Psi_{0,v}(v)s(o)dP_{0}(o)\\
&=
\int (\Psi_{0,v}(v)-\widetilde{\Psi}_0)s(o) dP_{0}(o).
\end{align*}

Combining the two terms, we obtain
\begin{align*}
\frac{\partial}{\partial t}\widetilde{\Psi}_t\Bigg|_{t=0}
&=
\int D^{v}_{P_0}(o)s(o)\,dP_0(o)
+
\int (\Psi_{0,v}(v)-\widetilde{\Psi}_0)s(o)\,dP_0(o)\\
&=
\int \Big(D^{v}_{P_0}(o)+\Psi_{0,v}(v)-\widetilde{\Psi}_0\Big)s(o)\,dP_0(o).
\end{align*}
The EIF of $\widetilde{\Psi}_{0}$ is 
\begin{align*}
\widetilde{D}_{P_0}(o)
=
D^{v}_{P_0}(o)
+
\Psi_{0,v}(v)
-
\widetilde{\Psi}_0.
\end{align*}
Equivalently,
\begin{align*}
D^{v}_{P_0}(o)
=
\widetilde{D}_{P_0}(o)
-
\Psi_{0,v}(v)
+
\widetilde{\Psi}_0.
\end{align*}

Substituting this into the general form of $D^*_{h,0}$ gives
\begin{align*}
D^*_{h,0}(o)
&=
\Big(
\widetilde{D}_{P_0}(o)
-
\Psi_{0,v}(v)
+
\widetilde{\Psi}_0
-
\widetilde{\Psi}_0
+
\Psi_{0,v}(v)
\Big)
h^c_{P_{0,V}}(v)
-
\Omega_{P_0}(h)\\
&=
\widetilde{D}_{P_0}(o)\,h^c_{P_{0,V}}(v)
-
\Omega_{P_0}(h).
\end{align*}
This completes the proof.
\end{proof}

\subsection{Proof of Theorem~\ref{theorem2}}
\begin{proof}
Recall that the pathwise differentiability of $\Omega_{P_{0}}(h)$ informs the following von Mises expansion 
\begin{align*}
     \Omega_{\widehat{P}_n}(h)- \Omega_{P_{0}}(h)=- P_0 D^*_{h,n} + R_{h}(\widehat{P}_{n},P_{0}) ,
\end{align*}
where $\widehat{P}_n$ is any estimator for $P_{0}$ and $R_{h}(\widehat{P}_{n},P_{0}) \coloneqq \Omega_{\widehat{P}_n}(h)- \Omega_{P_{0}}(h) + P_{0}D^{*}_{h,n} $ is a second order reminder term. We can then express our one-step estimator, $\widehat{\Omega}^{os}_{n}(h)$, as 
\begin{align*}
    \widehat{\Omega}^{os}_{n}(h)- \Omega_{P_{0}}(h)& = (\mathbb{P}_n-P_{0})D^*_{h,0} + (\mathbb{P}_n-P_{0})(D^*_{h,n}-D^*_{h,0}) + R_{h}(\widehat{P}_{n},P_{0}) \\
    & =  \mathbb{P}_nD^*_{h,0} + r_n(h),
\end{align*}
where $r_n(h)\coloneqq (\mathbb{P}_n-P_{0})(D^*_{h,n}-D^*_{h,0}) + R_{h}(\widehat{P}_{n},P_{0})$. Under the condition that $\underset{h \in\mathcal{H}}{\text{sup }} |r_n(h)|=o_{P}(n^{-1/2})$, uniform asymptotic linearity follow immediately. Furthermore, under the $P_0$-Donsker condition on $\mathcal{H}$, the empirical process
\{$\sqrt{n}(\mathbb{P}_n-P_0)D^*_{h,0}, h\in \mathcal{H}\}$ converges weakly to a tight mean-zero  Gaussian process $\mathbb{G}$ in $\ell^\infty(\mathcal{H})$, and weak convergence of 
$\{\sqrt{n}( \widehat{\Omega}^{os}_{n}(h)- \Omega_{P_{0}}(h)),h\in \mathcal{H}\}$ follows as a consequence (see for e.g., Chapter 2 of \cite{van1996weak}).
\end{proof}

\subsection{Discussion on the Donsker condition}
The Donsker conditions on $\mathcal{H}$ and $\{D^*_{h,0}, D^*_{h,n}: h \in \mathcal{H}\}$ in Theorem~\ref{theorem2} play distinct roles in establishing uniform convergence of $\{\sqrt{n}(\mathbb{P}_n - P_0)D^*_{h,0}: h \in \mathcal{H}\}$. We impose a Donsker condition on $\mathcal{H}$ to control the supremum norm, and a Donsker condition on $\{D^*_{h,0}, D^*_{h,n}: h \in \mathcal{H}\}$ to control the empirical process term $ (\mathbb{P}_n - P_0)(D^*_{h,n} - D^*_{h,0})$. In fact, a $P_0$-Donsker condition on $\mathcal{H}$ may be sufficient to imply that the class $\{D^*_{h,0}, D^*_{h,n}: h \in \mathcal{H}\}$ is also $P_0$-Donsker under a suitable Lipschitz-type condition, which can be verified directly in a given application. See Lemma~S1 in the Supplementary Material of \cite{hudson2026inference} for details. Thus, it may not be necessary to impose separate Donsker assumptions on the induced classes $\{D^*_{h,0}: h \in \mathcal{H}\}$ and $\{D^*_{h,n}: h \in \mathcal{H}\}$.

\section{Details on the illustrative examples}\label{sec:details-examples}
In this section, we provide detailed analysis of our proposed one-step estimator in Examples~\ref{examp1}-\ref{examp3} and show conditions under which they are uniformly asymptotically linear. Recall that we can express our one-step estimator, $\widehat{\Omega}^{os}_{n}(h)$, as 
\begin{align*}
    \widehat{\Omega}^{os}_{n}(h)- \Omega_{P_{0}}(h)
    & =  \mathbb{P}_nD^*_{h,0} + r_n(h)
\end{align*}
where $r_n(h)\coloneqq (\mathbb{P}_n-P_{0})(D^*_{h,n}-D^*_{h,0}) + R_{h}(\widehat{P}_{n},P_{0})$ and $ R_{h}(\widehat{P}_{n},P_{0})$  is defined in \eqref{secondorder}. For $\widehat{\Omega}^{os}_{n}(h)$ to be uniformly asymptotically linear, we require that $\underset{h \in \mathcal{H}}{\text{sup }} |r_n(h)|=o_{P}(n^{-1/2})$. Sections~\ref{illustration} and \ref{sec:additional-examples} present the form of $\widehat{\Omega}^{os}_{n}(h)$ for Examples~\ref{examp1}--\ref{examp3}. In this section, we describe conditions under which $\underset{h \in \mathcal{H}}{\text{sup }} |r_n(h)|=o_{P}(n^{-1/2})$ hold. Specifically, we will show that  $\underset{h \in \mathcal{H}}{\text{sup }}|(\mathbb{P}_n-P_{0})(D^*_{h,n}-D^*_{h,0})|=o_{P}(n^{-1/2})$ and $\underset{h \in \mathcal{H}}{\text{sup }} |R_{h}(\widehat{P}_{n},P_{0})|=o_{P}(n^{-1/2})$. To justify  $\underset{h \in \mathcal{H}}{\text{sup }}|(\mathbb{P}_n-P_{0})(D^*_{h,n}-D^*_{h,0})|=o_{P}(n^{-1/2})$, it is sufficient to show that $\underset{h \in \mathcal{H}}{\text{sup }} |P_{0}[D^*_{h,n}-D^*_{h,0}]^{2}|=o_p(1)$ in addition to assumption that $D^*_{h,n}$ and $D^*_{h,0}$ belong to some $P_0$-Donsker class (see Lemma~\ref{lemma1}). 

\begin{lemma}\label{lemma1}
Let $\mathbb{G}_n f=\sqrt{n}(\mathbb{P}_n-P_{0})f$. Suppose $\mathcal F$ is a fixed $P_{0}$-Donsker class such that, with probability tending to one,
\begin{align*}
    \{f_{n,h}:h\in\mathcal H\}\cup\{f_{0,h}:h\in\mathcal H\}\subset \mathcal F, \text{ and } \sup_{h\in\mathcal H}\|f_{n,h}-f_{0,h}\|_{L^{2}(P_{0})}=o_p(1), 
\end{align*}
then
\begin{align*}
\sup_{h\in\mathcal H}\left|\mathbb G_n (f_{n,h}- f_{0,h})\right|=o_p(1).
\end{align*}
\end{lemma}
\begin{proof}
Since $\mathcal F$ is $P_0$-Donsker, $\mathbb G_n$ converges weakly in $\ell^{\infty}(\mathcal F)$. Hence, by Theorems 1.5.4 and 1.5.7 of \cite{van1996weak}, the empirical process $\mathbb G_n$ is asymptotically uniformly $\rho$-equicontinuous on $\mathcal F$, where
\[
\rho(f,g)=\|f-g\|_{L^2(P_0)}.
\]
Thus, for any $\varepsilon,\eta>0$, there exists $\delta>0$ such that
\begin{align*}
\limsup_{n\to\infty}
P_0\Big(
\sup_{\substack{f,g\in\mathcal F\\ \|f-g\|_{L^2(P_0)}\le \delta}}
|\mathbb G_n f-\mathbb G_n g|>\varepsilon
\Big)
<\eta.
\end{align*}
Let $\Delta_n:=\underset{{h\in\mathcal H}}{\sup}\|f_{n,h}-f_{0,h}\|_{L^2(P_0)}$.
Then
\begin{align*}
&P_0\Big(
\sup_{h\in\mathcal H}|\mathbb G_n(f_{n,h}-f_{0,h})|>\varepsilon
\Big) =
P_0\Big(
\sup_{h\in\mathcal H}|\mathbb G_n f_{n,h}-\mathbb G_n f_{0,h}|>\varepsilon
\Big) \\
&\le
P_0(\Delta_n>\delta)
+
P_0\Big(
\sup_{\substack{f,g\in\mathcal F\\ \|f-g\|_{L^2(P_0)}\le \delta}}
|\mathbb G_n f-\mathbb G_n g|>\varepsilon
\Big).
\end{align*}
The first term converges to zero by assumption, and the second term is asymptotically smaller than $\eta$. Since $\eta>0$ is arbitrary, it follows that
\[
\sup_{h\in\mathcal H}\Big|\mathbb G_n(f_{n,h}-f_{0,h})\Big|=o_p(1).
\]
\end{proof}

\noindent Now recall that 
\begin{align}
    R_{h}(\widehat{P}_{n},P_{0}) \coloneqq \Omega_{\widehat{P}_n}(h)- \Omega_{P_{0}}(h) + P_{0}D^{*}_{h,n}.\label{secondorder}
\end{align}
In each of the examples, we study $R_{h}(\widehat{P}_{n},P_{0})$, and $P_{0}[D^*_{h,n}-D^*_{h,0}]^{2}$ and show conditions under which $\underset{h \in \mathcal{H}}{\text{sup }} \{P_{0}[D^*_{h,n}-D^*_{h,0}]^{2}\}^{1/2}=o_p(1)$ and $\underset{h \in \mathcal{H}}{\text{sup }} |R_{h}(\widehat{P}_{n},P_{0})|=o_{P}(n^{-1/2})$. 

To establish that $\underset{h \in\mathcal{H}}{\text{sup }} |r_n(h)|=o_{P}(n^{-1/2})$, it suffices to show that $\underset{h \in \mathcal{H}}{\text{sup }}|(\mathbb{P}_n-P_{0})(D^*_{h,n}-D^*_{h,0})|=o_{P}(n^{-1/2})$ and $\underset{h \in \mathcal{H}}{\text{sup }} |R_{h}(\widehat{P}_{n},P_{0})|=o_{P}(n^{-1/2})$. For Examples~\ref{examp1}--\ref{examp3}, we indicate conditions for which $\underset{h \in \mathcal{H}}{\text{sup }}|(\mathbb{P}_n-P_{0})(D^*_{h,n}-D^*_{h,0})|=o_{P}(n^{-1/2})$ and $\underset{h \in \mathcal{H}}{\text{sup }} |R_{h}(\widehat{P}_{n},P_{0})|=o_{P}(n^{-1/2})$ hold.

\subsection{Example 1 }
In this example, our function-valued parameter is  $\Psi_{0,x}=\E[Y|X=x]$. Denote $\bar{y}_{n}=n^{-1}\sum_{i=1}^nY_i$, $\bar{y}_{0}=\E[Y]$, $\bar{h}_{n}= \frac{1}{n} \sum_{i=1}^n h(X_i)$ and $\bar{h}_{0}=\E[h(X)]$. We recall that 
\begin{align*}
    \Omega_{P_{0}}(h)&=\int \E[Y|X=x](h(x)-\bar{h}_{0})dP_{0,X}(x) \\
  D^*_{h,0}(o) &= (y -\bar{y}_{0})
(h(x)-\bar{h}_{0})-  \Omega_{P_0}(h) \\
  D^*_{h,n}(o) &= (y -\bar{y}_{n})
(h(x)-\bar{h}_{n})-  \Omega_{\widehat{P}_n}(h). 
\end{align*}
We study $R_{h}(\widehat{P}_{n},P_{0})$:
\begin{align*}
    R_{h}(\widehat{P}_{n},P_{0})&= \int (y-\bar{y}_{n})(h(x)-\bar{h}_{n})dP_{0}- \int \E[Y|X=x](h(x)-\bar{h}_{0})dP_{0,X}(x)\\
    &= \int (yh(x) -y\bar{h}_{n} -\bar{y}_nh(x) +\bar{y}_n\bar{h}_{n} -yh(x) +\bar{y}_{0}\bar{h}_{0})dP_{0}(o) \\
    &=  \int ( -y\bar{h}_{n} -\bar{y}_nh(x) +\bar{y}_n\bar{h}_{n}+\bar{y}_{0}\bar{h}_{0})dP_{0}(o) \\
    &= (\bar{y}_{0}-\bar{y}_{n}) (\bar{h}_n-\bar{h}_{0}). 
\end{align*}

Now consider 
\begin{align*}
    \underset{h \in \mathcal{H}}{\sup}  |R_{h}(\widehat{P}_{n},P_{0})| & = \underset{h \in \mathcal{H}}{\sup} |(\bar{y}_{0}-\bar{y}_{n}) (\bar{h}_n-\bar{h}_{0})| \\
    & = |(\bar{y}_{0}-\bar{y}_{n})| \underset{h \in \mathcal{H}}{\sup} | (\bar{h}_n-\bar{h}_{0})| \\ 
    &=  |(\bar{y}_{0}-\bar{y}_{n})| \frac{1}{\sqrt{n}}\underset{h \in \mathcal{H}}{\sup} | (\sqrt{n}(\mathbb{P}_{n}-P_{0})h)| \\
    &= o_{P}(n^{-1/2})
\end{align*}
Since $\bar{y}_{0}-\bar{y}_n=o_{P}(n^{-1/2})$ and $ \frac{1}{\sqrt{n}}\underset{h \in \mathcal{H}}{\sup} | (\sqrt{n}(\mathbb{P}_{n}-P_{0})h)|=o_{P}(n^{-1/2})$.
\\
\noindent Next, we study $[P_{0}[D^*_{h,n}-D^*_{h,0}]^{2}]^{1/2}$. Since $(\mathbb{P}_n-P_{0})(\Omega_{\widehat{P}_{n}}(h)-\Omega_{P_{0}}(h))=0$, we only focus on studying
\begin{align}
P_{0}[((y -\bar{y}_{n})
(h(x)-\bar{h}_{n})-(y -\bar{y}_{0})
(h(x)-\bar{h}_{0}))]^{2}. \label{empiMean}
\end{align}
Consider,
\begin{align*}
   &  \int [(y-\bar{y}_n)(h(x)-\bar{h}_{n})  - (y-\bar{y}_0)(h(x)-\bar{h}_{0})]^2dP_{0}(o) \\
   &  = \int [ y-\bar{y}_n)(h(x)-\bar{h}_{n})  - (y-\bar{y}_0)(h(x)-\bar{h}_{0}) + (y-\bar{y}_0)(h(x)-\bar{h}_n)- (y-\bar{y}_0)(h(x)-\bar{h}_n)]^{2} P_{0}(o)\\
   & = \int [-(h(x)-\bar{h}_n)(\bar{y}_n-\bar{y}_0) - (y-\bar{y}_0)(\bar{h}_n-\bar{h}_0)]^2dP_{0}(o) \\
    &=   \int [-(h(x)-\bar{h}_{0}) (\bar{y}_n-\bar{y}_0) + (\bar{h}_{0}-\bar{h}_n)(\bar{y}_n-\bar{y}_0) - (y-\bar{y}_0)(\bar{h}_n-\bar{h}_0)]^2dP_{0}(o).
\end{align*}
Taking square root and applying triangular inequality, we are left to study each of the terms below:
\begin{align*}
    \|(h(x)-\bar{h}_{0})\|_{L^{2}(P_{0})} |\bar{y}_n-\bar{y}_0|,\quad  
     \|\bar{h}_{0}-\bar{h}_n\|_{L^{2}(P_{0})}|(\bar{y}_n-\bar{y}_0)|, \text{ and } 
      \|\bar{h}_n-\bar{h}_0\|_{L^{2}(P_{0})}  \|y-\bar{y}_0\|_{L^{2}(P_{0})}.
\end{align*}
If the following conditions,
\begin{align*}
    \sup_{h \in \mathcal{H}} |\bar{h}_n - \bar{h}_0| = o_p(1), & \quad  |\bar{y}_n - \bar{y}_0| = O_p(n^{-1/2}), \\
    \sup_{h \in \mathcal{H}} \|h - \bar{h}_0\|_{L^2(P_0)} < \infty, & \text{ and } \text{Var}_{P_0}(y) < \infty
\end{align*}
 hond, we have that  $\underset{h \in \mathcal{H}}{\text{sup }}|(\mathbb{P}_n-P_{0})(D^*_{h,n}-D^*_{h,0})|=o_{P}(n^{-1/2})$ and hence conclude  $\underset{h \in \mathcal{H}}{\text{sup }} |r_n(h)|=o_{P}(n^{-1/2}).$

\subsection{Example 2}
In this example, our function-valued parameter is $\Psi_{0,x_{s}}=\E[\mu_{0}(1,X)-\mu_{0}(0,X)|X_{s}=x_{s}]$. Denote $\bar{h}_{0}=\E[h(X)]$, $\mu_{0}(t,x)=\E[Y|T=t,X=x]$, $\pi_{0}(t|x)=\text{Pr}(T=t|X=x)$, $\psi_{0}(o)= \frac{2t-1}{\pi_{0}(t|x)}\{y-\mu_{0}(t,x)\} +\mu_{0}(1,x) -\mu_{0}(0,x)$, and $\tau_{0}=\int \E[\mu_{0}(1,X)  - \mu_{0}(0,X)|X_{s}=x_{s}]dP_{0,X_{s}}(x_{s})$  with their corresponding estimators  $\bar{h}_{n}= \frac{1}{n} \sum_{i=1}^n h(X_i)$, $\mu_{n}(t,x)$, $\pi_{n}(t|x)$, $\psi_{n}(o)$, and $\tau_{n}$, respectively. Recall that  
\begin{align*}
 \Omega_{P_{0}}(h)& =  \int \E[\mu_{0}(1,x)-\mu_{0}(0,x)|X_{s}] (h(x_{s})-\bar{h}_{0})dP_{0,X_{s}}( x_{s}) \\
     D^{*}_{h,0}(o)&=\Big (\frac{2t-1}{\pi_{0}(t|x)}\{y-\mu_{0}(t,x)\} +\mu_{0}(1,x) -\mu_{0}(0,x) +\tau_{0}  \Big) (h(x_{s})-\bar{h}_{0})  -  \Omega_{P_{0}}(h) \\
       D^{*}_{h,n}(o)&=\Big (\frac{2t-1}{\pi_{0}(t|x)}\{y-\mu_{n}(t,x)\} +\mu_{n}(1,x) -\mu_{n}(0,x) +\tau_{n} \Big) (h(x_{s})-\bar{h}_{n}) - \Omega_{\widehat{P}_{n}}(h).
\end{align*}
We study $R_{h}(\widehat{P}_{n},P_{0})$:
\begin{align*}
R_{h}(\widehat{P}_n,P_{0}) &= \int \big\{ \psi_n(o) - \tau_n \big\} \big\{ h(x_{s})-\bar{h}_{n} \big\} \, dP_0(o) 
- \int \big\{ \psi_{0}(o) - \tau_0 \big\} \big\{ h(x_{s})-\bar{h}_{0}\big\} \, dP_0(o) \\
&= \int \big[ \psi_n(o) - \psi_{0}(o) \big] h(x_s) \, dP_0(o) -(\tau_{n}-\tau_{0})\bar{h}_{0}  \\
& \quad - \left[\int \{\psi_{n}(o)-\tau_{n}\}dP_{0}(o)   \right]\bar{h}_{n} + \left[\int \big\{ \mu_0(1, x) - \mu_0(0, x) - \tau_0 \big\} dP_{0}(o)\right] \bar{h}_{0} \\
&= \int \big[ \psi_n(o) - \psi_{0}(o) \big] h(x_s) \, dP_0(o) -(\tau_{n}-\tau_{0})\bar{h}_{0}  \\
& \quad - \left[\int \{\psi_{n}(o)-\tau_{n}\}dP_{0}(o)   \right]\bar{h}_{0} - \left[\int  \{\psi_{n}(o)-\tau_{n}\} dP_{0}(o)\right](\bar{h}_{n}- \bar{h}_{0})\\
&= \int \big[ \psi_n(o) - \psi_{0}(o) \big] h(x_s) \,dP_0(o) \\
& \quad - \left[\int \{\psi_{n}(o)-\tau_{0}\}dP_{0}(o)   \right]\bar{h}_{0} - \left[\int  \{\psi_{n}(o)-\tau_{n}\} dP_{0}(o)\right](\bar{h}_{n}- \bar{h}_{0})\\
&= \int \big[ \psi_n(o) - \psi_{0}(o) \big] h(x_s) \,dP_0(o) \\
& \quad - \left[\int \{\psi_{n}(o)-\tau_{0}\}dP_{0}(o)   \right]\bar{h}_{0} - \left[\int  \{\psi_{n}(o)-\tau_{n}\} dP_{0}(o)\right](\bar{h}_{n}- \bar{h}_{0})\\
& \quad + \left[ \int  \{ \mu_0(1, x) - \mu_0(0, x) -\tau_{0}\} dP_{0}(o) \right](\bar{h}_{n}- \bar{h}_{0})\\
&= \int \big[ \psi_n(o) - \psi_{0}(o) \big] h(x_s) \,dP_0(o) \\
& \quad - \left(\int  \psi_n(o) - \psi_{0}(o) dP_0(o) \right)\bar{h}_{0} - \left[\int  \{\psi_{n}(o)-\tau_{n}\} dP_{0}(o)\right](\bar{h}_{n}- \bar{h}_{0})\\
& \quad + \left[ \int  \{ \psi_{0}(o) -\tau_{0}\} dP_{0}(o) \right](\bar{h}_{n}- \bar{h}_{0}) \\
&= \int \big[ \psi_n(o) - \psi_{0}(o) \big] (h(x_s)-\bar{h}_{0}) \,dP_0(o) \\
& \quad + \left(\int  \psi_n(o) - \psi_{0}(o) dP_0(o) \right)(\bar{h}_{n}-\bar{h}_{0} )\\
& \quad + (\tau_{n}-\tau_{0})(\bar{h}_{n}-\bar{h}_{0}).
\end{align*}
If the following conditions
\begin{align*}
         \int  (\mu_{0}(t,x)-\mu_{n}(t,x))^2dP_{0}(o)= o_{P}(1), & \quad  \int (\pi_{0}(t,x)-\pi_{n}(t,x))^2dP_{0}(o)  =  o_{P}(1), \\
        \left [\int (\pi_{0}(t,x)-\pi_{n}(t,x))^2 dP_{0}(o)\right]^{1/2} & \left [ \int  (\mu_{0}(t,x)-\mu_{n}(t,x))^2dP_{0}(o) \right ]^{1/2}  = o_{P}(n^{-1/2}),   \\
    \sup_{h \in \mathcal{H}} |\bar{h}_n - \bar{h}_0| = o_P(1),\text{ and }  & \sup_{h \in \mathcal{H}} \|h - \bar{h}_0\|_{L^2(P_0)} < \infty 
\end{align*}
hold, we have that $\underset{h \in \mathcal{H}}{\sup}|R_{h}(\widehat{P}_{n},P_{0})|=o_{P}(n^{-1/2})$. 
\\ \\
\noindent Studying $\underset{h \in \mathcal{H}}{\sup} \{P_{0}[D^*_{h,n}-D^*_{h,0}]^{2}\}^{1/2} $ reduces to showing that
\begin{align*}
    & \underset{h \in \mathcal{H}}{\sup} \Big\{\int \Big(\big\{ \psi_n(o) - \tau_n \big\} \big\{ h(x_{s})-\bar{h}_{n} \big\} - \big\{ \psi_{0}(o) - \tau_0 \big\} \big\{ h(x_{s})-\bar{h}_{0}\big\}\Big)^{2} dP_0(o)\Big\}^{1/2} =o_{P}(1).
\end{align*}
Following the rate conditions established earlier, we have $underset{h \in \mathcal{H}}{\sup} \{P_{0}[D^*_{h,n}-D^*_{h,0}]^{2}\}^{1/2}=o_{P}(1)$, which implies $ \underset{h \in \mathcal{H}}{\sup}|(\mathbb{P}_n-P_{0})(D^*_{h,n}-D^*_{h,0})|=o_{P}(n^{-1/2})$. Combining all the above, we conclude that $\underset{h \in \mathcal{H}}{\sup} |r_n(h)|=o_{P}(n^{-1/2})$. The details presented here have also been studied in \cite{dukes2024nonparametric}.

\subsection{Example 3}
Our function-valued parameter is $\Psi_{0,z}=\E[(Y-\E[Y|Z=z])(X-\E[X|Z=z])|Z=z]$. Let $\mu_{n,y|z}$ and $\mu_{n,x|z}$ be estimators of $\mu_{0,y|z}\coloneq \E[Y|Z=z]$ and $\mu_{0,x|z}\coloneq\E[X|Z=z]$ respectively. Recall that 
\begin{align*}
 \Omega_{P_{0}}(h)& =\int Cov(X,Y|Z=z)(h(z)-\bar{h}_{0})dP_{0,Z}(z) \\
       D^{*}_{h,0}(o)&=\Big ((y-\mu_{0,y|z})(x-\mu_{0,x|z})-\int Cov(X,Y|Z=z) dP_{0,z}\Big)h^{c}_{P_{0,Z}}(z) - \Omega_{P_{0}}(h) \\
        D^{*}_{h,n}(o)&=\Big ((y-\mu_{n,y|z})(x-\mu_{n,x|z})-\mathbb{P}_{n}[(y-\mu_{n,y|z})(x-\mu_{n,x|z})]\Big)h^{c}_{\mathbb{P}_{n,Z}}(z) - \Omega_{\widehat{P}_{n}}(h)
   \end{align*}

We study $R_{h}(\widehat{P}_{n},P_{0})$:
\begin{align*}
    R_{h}(\widehat{P}_{n},P_{0})
    & = \int (y-\mu_{n,y|z})(x-\mu_{n,x|z})(h(z)-\bar{h}_{n})dP_{0}(o) \\
    & - \int \mathbb{P}_{n}[(y-\mu_{n,y|z})(x-\mu_{n,x|z})](h(z)-\bar{h}_{n})dP_{0}(o) \\
    & - \int (y-\mu_{0,y|z})(x-\mu_{0,x|z})(h(z)-\bar{h}_{0})dP_{0}(o)\\
    &=  \int\big(  (y-\mu_{n,y|z})(x-\mu_{n,x|z})-(y-\mu_{0,y|z})(x-\mu_{0,x|z})\big)h(z)dP_{0}(o)\\
    & - \int\big(  (y-\mu_{n,y|z})(x-\mu_{n,x|z})-(y-\mu_{0,y|z})(x-\mu_{0,x|z})\big)\bar{h}_{0}dP_{0}(o)\\
    & - \int  (y-\mu_{n,y|z})(x-\mu_{n,x|z})(\bar{h}_{n}-\bar{h}_0)dP_{0}(o)\\
    & - \mathbb{P}_{n}[(y-\mu_{n,y|z})(x-\mu_{n,x|z})](\bar{h}_{0}-\bar{h}_{n})\\
    &= \int (\mu_{n,y|z}-\mu_{0,y|z})(\mu_{n,x|z}-\mu_{0,x|z})h(z)dP_{0,Z}(z) \\
    & - \int(\mu_{n,y|z}-\mu_{0,y|z})(\mu_{n,x|z}-\mu_{0,x|z})\bar{h}_{0}dP_{0,Z}(z) \\
    & +  (\bar{h}_{0}-\bar{h}_{n})\int(\mu_{n,y|z}-\mu_{0,y|z})(\mu_{n,x|z}-\mu_{0,x|z})dP_{0,Z}(z)\\
    & + (\bar{h}_{0}-\bar{h}_{n})\int (y-\mu_{0,y|z})(x-\mu_{0,x|z})dP_{0}(o) \\
    & - \mathbb{P}_{n}[(y-\mu_{n,y|z})(x-\mu_{n,x|z})](\bar{h}_{0}-\bar{h}_{n})\\
    &= \int (\mu_{n,y|z}-\mu_{0,y|z})(\mu_{n,x|z}-\mu_{0,x|z})(h(z)-\bar{h}_{0})dP_{0,Z}(z) \\
    & +  (\bar{h}_{0}-\bar{h}_{n})\int(\mu_{n,y|z}-\mu_{0,y|z})(\mu_{n,x|z}-\mu_{0,x|z})dP_{0,Z}(z) \\ 
    & + (\bar{h}_{0}-\bar{h}_{n})\Big ( \int (y-\mu_{0,y|z})(x-\mu_{0,x|z})dP_{0}(o) - \mathbb{P}_{n}[(y-\mu_{n,y|z})(x-\mu_{n,x|z})]  \Big ).
\end{align*}
Hence, if both $(\mu_{n,y|z}-\mu_{0,y|z})$ and $(\mu_{n,x|z}-\mu_{0,x|z})$ converge to zero in $L_{2}(P_{0})$ norm at the rate $o_{P}(n^{-1/4})$ in addition to $\sup_{h \in \mathcal{H}} |\bar{h}_n - \bar{h}_0| = o_P(1)$ and  $ \sup_{h \in \mathcal{H}} \|h - \bar{h}_0\|_{L^2(P_0)} < \infty$, we have $\underset{h in \mathcal{H}}{\sup}|R_{h}(\widehat{P}_{n},P_{0})=o_{P}(n^{-1/2})|$. In fact, the $L_{2}(P_{0})$ norm condition can generally be satisfied if the following hold 
        \begin{align*}
       \left [\int  (\mu_{0,y|z}-\mu_{n,y|z})^2dP_{0}\right]^{1/2} \left [ \int (\mu_{0,x|z}-\mu_{n,x|z})^2dP_{0} \right ]^{1/2} = o_{P}(n^{-1/2}) .        \end{align*}

\noindent Studying $\underset{h \in \mathcal{H}}{\sup}\{P_{0}[D^*_{h,n}-D^*_{h,0}]^{2}\}^{1/2}$ reduces to showing that
\begin{align*}
    &\int \Big \{(y-\mu_{n,y|z})(x-\mu_{n,x|z})
    -\Pp_{n}[(y-\mu_{n,y|z})(x-\mu_{n,x|z})](h(z)-\bar{h}_{n}) \\
    & -  ((y-\mu_{0,y|z})(x-\mu_{0,x|z}) -\int (x-\mu_{0,x|z})(y-\mu_{0,y|z}dP_{0})(h(z)-\bar{h}_{0}) \Big \}^{2}dP_{0}(o) \\
    &=o_{P}(1).
\end{align*}
Following the rate conditions established earlier, we have $\underset{h \in \mathcal{H}}{\sup} \{P_{0}[D^*_{h,n}-D^*_{h,0}]^{2}\}^{1/2}=o_{P}(1)$, which implies $ \underset{h \in \mathcal{H}}{\sup}|(\mathbb{P}_n-P_{0})(D^*_{h,n}-D^*_{h,0})|=o_{P}(n^{-1/2})$. Combining all the above, we conclude that $\underset{h \in \mathcal{H}}{\sup} |r_n(h)|=o_{P}(n^{-1/2})$.

\end{document}